\documentclass[journal,comsoc]{IEEEtran}

\usepackage{amsmath,graphics,amssymb,epsfig}
\usepackage{subcaption}
\usepackage{url}
\usepackage{color,soul}
\usepackage{verbatim}
\usepackage{amsthm}
\usepackage{tabularx}
\usepackage{cite}
\usepackage{multirow}

\usepackage{bbm}
\usepackage{algorithm}
\usepackage{algpseudocode}
\usepackage{graphicx}
\usepackage{textcomp}
\usepackage{soul}
\usepackage{svg}
\usepackage{booktabs}
\usepackage{hyperref}
\usepackage{makecell}

\usepackage{tikz}
\usetikzlibrary{shapes.geometric}

\newtheorem{remark}{Remark}
\newtheorem{proposition}{Proposition}

\newcommand{\Input}{\State \textbf{Input:} }
\newcommand{\Output}{\State \textbf{Output:} }

\DeclareMathOperator*{\argmax}{argmax}

\begin{document}
\setlength{\abovedisplayskip}{5pt}
\setlength{\belowdisplayskip}{5pt}

\title{Generating High Dimensional User-Specific Wireless Channels using Diffusion Models}

\author{\IEEEauthorblockN{Taekyun~Lee, Juseong~Park, Hyeji~Kim, and~Jeffrey~G.~Andrews} \\ 
\thanks{This work was partly supported by NSF Award CNS-2148141 and Keysight Technologies through the 6G@UT center within the Wireless Networking and Communications Group (WNCG) at the University of Texas at Austin, as well as ARO Award W911NF2310062 and ONR Award N000142412542.}
\thanks{Taekyun Lee, Juseong Park, Hyeji Kim, and Jeffrey G. Andrews are with the 6G@UT center in the WNCG at the University of Texas at Austin, Austin, TX 78712, USA (email: taekyun@utexas.edu, juseong.park@utexas.edu, hyeji.kim@austin.utexas.edu, jandrews@ece.utexas.edu). A preliminary version appeared in \cite{lee2024asilomar}.}
\thanks{Last revised: \today.
}}
\maketitle


\begin{abstract}

Deep neural network (DNN)-based algorithms are emerging as an important tool for many physical and MAC layer functions in future wireless communication systems, including for large multi-antenna channels.  However, training such models typically requires a large dataset of high-dimensional channel measurements, which are very difficult and expensive to obtain. This paper introduces a novel method for generating synthetic wireless channel data using diffusion-based models to produce user-specific channels that accurately reflect real-world wireless environments. Our approach employs a conditional denoising diffusion implicit model (cDDIM) framework, effectively capturing the relationship between user location and multi-antenna channel characteristics. 
We generate synthetic high fidelity channel samples using user positions as conditional inputs, creating larger augmented datasets to overcome measurement scarcity.  The utility of this method is demonstrated through its efficacy in training various downstream tasks such as channel compression and beam alignment. 
Our diffusion-based augmentation approach achieves over a 1-2 dB gain in NMSE for channel compression, and an 11 dB SNR boost in beamforming compared to prior methods, such as noise addition or the use of generative adversarial networks (GANs).

\end{abstract}


\section{Introduction}

\IEEEPARstart{M}ASSIVE multiple-input multiple-output (MIMO) is a foundational technology for the lower bands in 5G and will continue to evolve to larger dimensional channels in 6G, for example in the upper midband channels above 7 GHz \cite{Zhang24, FCCtac:23}.    Accurately measuring or estimating high-dimensional Channel State Information (CSI) is challenging and costly in terms of energy and bandwidth since many pilot tones are needed to determine the entire channel matrix, which is also frequency and time-varying.  Thus, even achieving accurate receiver-side CSI is nontrivial.    Recent work has shown that deep neural network (DNN)-based approaches have the potential to excel in the high-dimensional MIMO regime for tasks including detection \cite{He:20}, precoding \cite{Soh:21,park:24}, channel estimation \cite{doshi2022over, Arv:23}, and channel compression \cite{CsiNet, CsiNetPlus}.  This is in part due to their lack of reliance on a high-quality channel estimate, and to their ability to exploit learned structure in the underlying MIMO channel.  

Meanwhile, the larger antenna arrays at the higher carrier frequencies such as millimeter wave (mmWave) being considered for future systems, will create a higher dimensional channel. Although the mmWave channel rank is typically small, and the main goal is to find a high-SNR beam direction or beam pair, this task is also known to be challenging and slow, and deep learning methods show great promise in reducing the complexity and latency of beam alignment \cite{heng:23, Raj22}.
However, such deep learning methods still typically require many full-dimensional (i.e. the whole channel matrix $\mathbf{H}$) measurements, that are site-specific, to properly train the models.


In short, a common challenge of DNN-based MIMO methods---at both the lower and upper frequency bands---is their reliance on vast amounts of site-specific channel data. Specifically, training a beam alignment algorithm typically requires on the order of 100K channel measurements for each macrocell sector \cite{Raj22, heng:23, alkhateeb:19, Wang22}.   When dealing with high-dimensional channel measurements, the complexity escalates further due to the need to track a vast number of parameters at the receiver for accurate channel estimation \cite{Eli17}.  In systems like massive MIMO, this complexity is compounded by the limited number of RF chains available at the base station, leading to excessively high pilot overheads \cite{Ma20}.  Needless to say, it is very time-consuming and prohibitively expensive to collect on the order of 100K physical $\mathbf{H}$ measurements in each cell site.   Even if an offline ray tracing approach is used, which negates the need for field measurements, it is still necessary to carefully specify each cell's physical environment and accurately model its propagation characteristics, e.g. reflection coefficients of each object.  To harness the power of DNN-based approaches to MIMO systems, there is a pressing need to develop rich and realistic high-dimensional channel datasets despite having minimal actual data. Crucially, additional data must meet the necessary condition of preserving the channel's spatially correlated multipath structure to ensure reliable performance \cite{Rappaport2013mmWave5G}.


\subsection{Background \& Related works}

\smallskip \textbf{Data augmentation} To overcome the challenges posed by the limited number of physical measurements or ray traced channel samples, a promising approach is to use data augmentation, in particular using deep learning models to perform the augmentation. Non-generative models like convolutional neural networks (CNNs) or autoencoders (AEs) can capture channel statistics and generate synthetic channels \cite{Sol20, Li:21}.  However, their structure cannot accurately represent higher dimensional distributions and are limited in their ability to generate diverse datasets. 


Turning to deep generative models, a generative adversarial network (GAN) can learn complicated channel distributions and generate channel matrices \cite{Xia:22, Liang:20, Yang:19, Bal2020DosJalDimAnd}. In particular, conditional GANs have generated channels for air-to-ground communication by outputting path gain and delay separately \cite{cGAN_a2g}. A study using denoising diffusion probabilistic models (DDPM) for dataset augmentation \cite{Xu23} focuses on the tapped delay line (TDL) dataset, differing from high-dimensional mmWave channels. Similarly, \cite{Sengupta23} utilizes DDPM for dataset augmentation, but specifically for MIMO channels, in contrast to the earlier work focused on single-input single-output (SISO).   A very recent work \cite{Baur24} generates synthetic channels using a diffusion model and evaluating downstream tasks with real-world measurements in a multiple-input single-output (MISO) channel setting, showing that different generative models excel in different downstream tasks.
However, these prior works \cite{Xu23, Sengupta23, Baur24} may not be well-suited for augmenting datasets in wireless downstream tasks. Since the generated data follows the prior channel distribution, low-probability channel instances are rarely generated. Consequently, downstream task DNNs are rarely exposed to such atypical cases during training, limiting their ability to generalize effectively and potentially leading to performance degradation. 
A very recent work \cite{newZhang24} employs DDPM for channel augmentation, but our approach uniquely conditions on UE positions, which significantly enhances performance in downstream tasks.
Indeed, generating diverse and broader coverage data for augmentation is a longstanding concept in both the deep learning \cite{hendrycks2020augmix, Tifrea2022, kang2020decoupling} and statistical \cite{Chawla2002} literature. 




\smallskip \textbf{Diffusion models.}  Our work focuses on diffusion models, which are one of the most powerful and recently proposed deep generative models \cite{DDPM}.   Diffusion models have separated superimposed sources in radio-frequency systems \cite{Tejas:23} by formulating statistical priors with new posterior estimation objective functions and employing a score-matching scheme for multi-source scenarios. The score-matching model was also used for channel estimation \cite{Arv:23}, demonstrating that it can outperform conventional compressed sensing methods for 3rd Generation Partnership Project (3GPP) channel models \cite{3gpp}.

\subsection{Contributions}

We propose a novel approach to MIMO channel data generation, which is \textit{diffusion model-based CSI augmentation}. Our approach uses a small number of true $\mathbf{H}$ measurements to generate a much larger set of augmented channel samples using a diffusion model.  Our approach specifically draws inspiration from diffusion autoencoders \cite{pre:22}, considering the user's position as a conditional input and employing a diffusion model as the decoder.
We combine classifier guidance \cite{classifierguide} and the denoising diffusion implicit models (DDIM) framework \cite{DDIM}, capturing the relationship between the user's position and its MIMO channel matrix.  Our main contributions can be summarized as follows. 




\textbf{Position-based generative models for channel data prediction.} 
    This paper is, to the best of our knowledge, the first to explicitly \emph{condition on the UE position} to generate channel samples. More specifically, our method employs a conditional diffusion model \cite{classifierguide} to learn the MIMO channel distribution and predict the expected channel data at a desired user equipment (UE) location.  Furthermore, we apply a discrete Fourier transform (DFT) to obtain the beamspace representation of the channel, which we find crucial for enabling the diffusion model to accurately capture the underlying distribution. By leveraging position-based data augmentation, our approach differs from methods sampling channels under a specific distribution, producing a broader coverage set of channel samples and thereby improving the diversity of training data more suited for downstream tasks.
    
\textbf{Necessity of the Diffusion Model Framework.}
Our experiments show that when the training and test distributions are the same, supervised training using the same backbone structure can yield good results. However, when the distributions differ, only the diffusion model framework—which captures the probability distribution—works effectively.

 \textbf{Validation of the proposed dataset augmentation method in wireless communication tasks.} Our experiments demonstrate the effectiveness of augmented datasets across two important wireless communication tasks: compressing CSI feedback in massive MIMO systems using CRNet \cite{CRNet},
    and site-specific grid-free beam adaptation \cite{heng:23}. 
    The results indicate that diffusion model-based augmentation outperforms other methods with few measurements, and achieves comparable performance to true channel data.  

    This paper significantly extends our preliminary work presented at the Asilomar Conference \cite{lee2024asilomar} by expanding the scope beyond CSI compression to include additional downstream applications, providing comprehensive out-of-distribution evaluation, and detailed theoretical analysis.
    

\subsection{Notation \& Organization}
The rest of the paper is organized as follows. Preliminaries of the system model and problem setup are explained in Section \ref{prob_st}, followed by a detailed description of the proposed methods in Section \ref{Proposed Methods}. In Section \ref{Visualization and Evaluation}, we visualize the generated channels and provide a quantitative analysis. Each downstream application and the results are explained in Section \ref{Downstream Applications and Evaluation}. The paper concludes with future directions in Section \ref{Conclusion}. 

\emph{Notation:} $\mathbf{A}$ is a matrix and $\mathbf{a}$ is a vector. A continuous process is $\mathbf{A}(t)$, while discrete is $\mathbf{A}[t]$ when $t$ represents time. $\left\|\mathbf{A}\right\|_F^2$ is the Frobenius norm of $\mathbf{A}$. Concatenating the columns of  
$\mathbf{A}$ column-wise, we obtain the vector $\underline{\mathbf{A}}$. 


\section{Preliminaries and Problem Statement}
\label{prob_st}

\begin{figure*}[!tbp]
  \centering
  \includegraphics[width=1.0\linewidth]{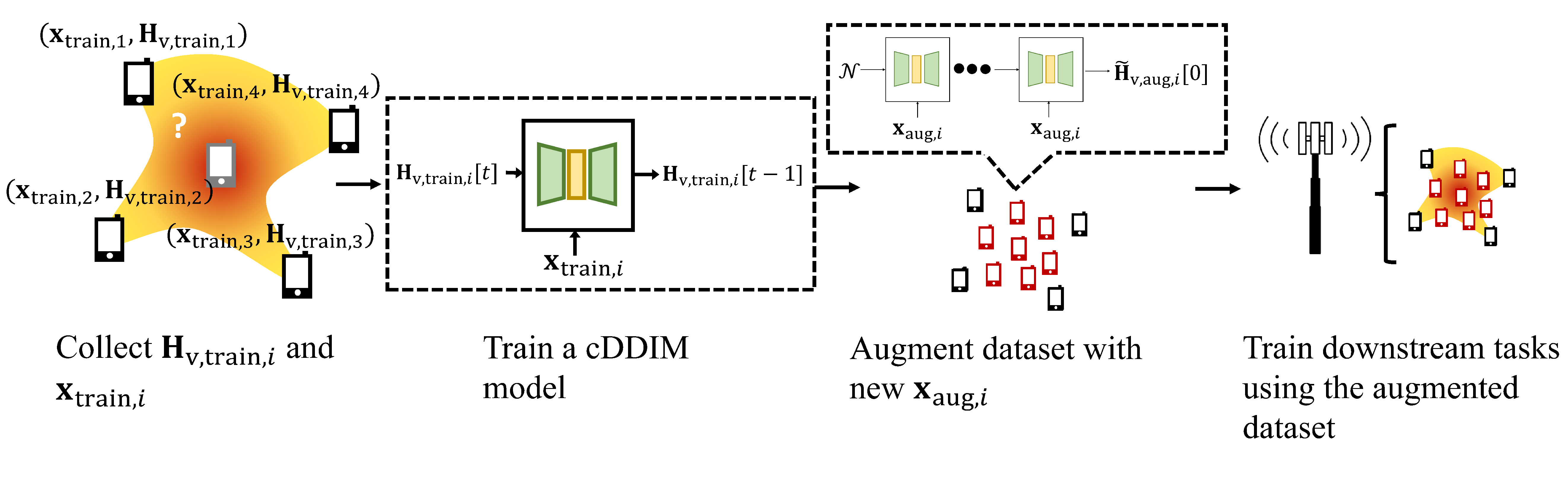}
  \caption{Illustration of the proposed wireless channel generation scenario.}\label{figure:proposedapproach.pdf}
\end{figure*}


\subsection{System Model (Channel Model)}
\label{system_model}
A narrowband massive MIMO system is considered, where a transmitter with $N_t$ antennas serves a receiver with $N_r$ antennas. This model can be extended in a straightforward way to wideband systems by incorporating frequency selectivity, such as subbands via block fading.

At 28~GHz with our $N_t=32$ ULA, any UE at $d\geq 20\ \text{m}$ is beyond the Fraunhofer distance $D_{\mathrm{FF}}\approx 5\ \text{m}$, so a planar-wave model is sufficient.
Assuming a 3D-channel model with $L$ propagation paths
\cite{say02}, the channel matrix ${\mathbf{H}} \in \mathbb{C}^{N_r\times N_t}$ can be expressed as
\begin{align}
{\mathbf{H}} = \sum\limits_{i = 1}^L {{\gamma_{i}}} {{\mathbf{a}}_r}\left( {\theta^{r} _{i}},{\phi^{r} _{i}} \right){\mathbf{a}}_t^H\left( {\theta^{t} _{i}},{\phi^{t} _{i}} \right)
\label{eqn:channel}
\end{align}
where $\gamma_{i}$ is the complex channel gain of the $i$th path, $( {\theta^{r} _{i}},{\phi^{r} _{i}})$ and $( {\theta^{t} _{i}},{\phi^{t} _{i}})$ denote the azimuth and elevation angle-of-arrival (AoA) pair and those of the angle-of-departure (AoD) pair, respectively. 
Here, ${{\mathbf{a}}_r}( \cdot)\in \mathbb{C}^{N_{r}\times 1}$ and ${\mathbf{a}}_t(\cdot)\in \mathbb{C}^{N_{t}\times 1}$ respectively account for the transmit and receive array response vectors, but we do not specify the array structure. Note that $\mathbf{H}$ can be described by the sum of channel paths where each path is a function of five parameters: its AoA pair $( {\theta^{r} _{i}},{\phi^{r} _{i}})$, AoD pair $( {\theta^{t} _{i}},{\phi^{t} _{i}})$, and channel gain $\gamma _{i}$. 
Depending on the surrounding environment, the value of $L$, along with the above five parameters, is determined to characterize the clusters and channel paths.

The variables $\gamma_{i}, \theta^{r}_{i}, \phi^{r}_{i}, \theta^{t}_{i}, \phi^{t}_{i}, i \in [1,\cdots, L]$ and $L$ can be modeled as conditional random variables given the user's 3-dimensional (3D) position $\mathbf{x} = [x_1, x_2, x_3]^T$, where $x_1$ and $x_2$ denote the user's planar position, and $x_3$ indicates their height. We assume the BS is fixed at the origin. We use the QuaDRiGa \cite{quadriga2023} and DeepMIMO \cite{alkhateeb:19} datasets in this experiment, which produce $\mathbf{H}$ based on $\mathbf{x}$ with the above parameters calculated implicitly. More details about the simulators will be presented later.

Even though we do not explicitly make any assumptions about the channel distribution $p(\mathbf{H})$, such as the dimensionality or sparsity of the channel, we will utilize the following key insight: the beamspace representation of mmWave MIMO channels exhibits high spatial correlation due to clustering. Therefore, we will focus on the beamspace representation $\mathbf{H}_\mathrm{v}$, defined as
\begin{equation}
\mathbf{H}_\mathrm{v} = {{\mathbf{A}}_r^{H}} \mathbf{H} \mathbf{A}_t
\end{equation}
where $\mathbf{A}_t \in \mathbb{C}^{N_t \times N_t}$ and $\mathbf{A}_r \in \mathbb{C}^{N_r \times N_r}$ are unitary DFT matrices. We output the beamspace matrix \(\mathbf{H}_{\mathrm v}\) because it already embeds the key path parameters—angles, delays, and gains—and all of our downstream networks are trained directly on channel matrices. Furthermore, since the channel itself changes with the transmit/receive array geometry, \(\mathbf{H}_{\mathrm v}\) implicitly carries information about the array architecture as well.

\subsection{Problem Setup}
\label{Problem Setup}

The primary objectives of this paper are twofold: (a) develop a model that estimates the beamspace channel matrix $\mathbf{H}_\mathrm{v}$ from the user's position $\mathbf{x}$ by implicitly determining the relevant parameters as explained in Section \ref{system_model}, and (b) to use this model to augment a channel measurement dataset, for the purpose of training deep learning downstream tasks.  The model we train for estimating the channel matrix $\mathbf{H}_\mathrm{v}$ from the user's position $\mathbf{x}$ utilizes conditional DDIM, which we refer to as the cDDIM model. The details of this model will be discussed in the next section. In this section, we focus on the framework for the second objective: the dataset augmentation problem.

\smallskip \textbf{Method Overview.\ } 
Suppose we have access to 
$N_\mathrm{train}$ pairs of position-channel measurements, labeled as $(\mathbf{x}_{\mathrm{train},i}, \mathbf{H}_{\mathrm{v,train},i})$, where $i \in \{1, \ldots, N_\mathrm{train}\}$. We aim to expand this dataset by randomly selecting $N_\mathrm{aug}$ positions 
$\mathbf{x}_{\mathrm{aug},i}$ for $i \in \{1, \ldots, N_\mathrm{aug}\}$, 
and generating $N_\mathrm{aug}$ estimated channels at those positions $\widetilde{\mathbf{H}}_{\mathrm{v,aug},i}$, where $i \in \{1, \ldots, N_\mathrm{aug}\}$. 
This results in an augmented dataset with $N_\mathrm{train} + N_\mathrm{aug}$ pairs of position and channel data.

\smallskip \textbf{Assumptions on Channel Measurements.\ } We treat CSI generated by statistical or ray-tracing simulators as equivalent to physical measurements. If available, one could use actual measurements in our framework without modification.   Throughout this paper, we assume noise-free samples, but if the CSI is noisy (e.g., estimated from pilots), our results will degrade in proportion to the noise variance.



\smallskip \textbf{Framework.\ } The steps of our proposed method, depicted in Fig.~\ref{figure:proposedapproach.pdf}, are: 

\begin{enumerate}
    \item Channel measurements collection: We collect channel measurements $\{\mathbf{H}_{\mathrm{v,train},i}\}_{i=1}^{N_\mathrm{train}}$ and UE positions $\{\mathbf{x}_{\mathrm{train},i}\}_{i=1}^{N_\mathrm{train}}$. 
    As mentioned before, we can collect channel measurements from pilot sequences. Additionally, we should send the UE position $\mathbf{x}_{\mathrm{train},i}$ from the UE to the base station (BS). In a non-standalone system, this can be easily sent through a lower-frequency side link since $\mathbf{x}_{\mathrm{train},i}$ is just a vector of three float numbers.
    \item Training of the cDDIM model: The cDDIM model, our generative model further explained in Section \ref{Proposed Methods}, is trained using the given measurements $\{\mathbf{H}_{\mathrm{v,train},i}\}_{i=1}^{N_\mathrm{train}}$ by adding Gaussian noise and learning the denoising procedure. UE positions $\{\mathbf{x}_{\mathrm{train},i}\}_{i=1}^{N_\mathrm{train}}$ are used as a conditional input to the model.
    \item Channel synthesis via the cDDIM model: We use the trained model to generate synthetic channel matrices $\{\widetilde{\mathbf{H}}_{\mathrm{v,aug},i}\}_{i=1}^{N_\mathrm{aug}}$ from $N_\mathrm{aug}$ UE positions $\{\mathbf{x}_{\mathrm{aug},i}\}_{i=1}^{N_\mathrm{aug}}$ that were not included in the training dataset. 
    \item Training downstream task deep learning models with both synthesized and measured channels: We use the combined set of training and augmented $N_\mathrm{train} + N_\mathrm{aug}$ channel matrices $\{\mathbf{H}_{\mathrm{v,train},i}\}_{i=1}^{N_\mathrm{train}} \bigcup \{\widetilde{\mathbf{H}}_{\mathrm{v,aug},i}\}_{i=1}^{N_\mathrm{aug}}$ for downstream tasks.
    \item Evaluation: This amplification allows us to obtain a much larger set of channel matrices for data-driven downstream tasks. 
    
\end{enumerate}
    One might consider directly comparing the augmented channels  $\{\widetilde{\mathbf{H}}_{\mathrm{v,aug},i}\}_{i=1}^{N_\mathrm{aug}}$ to ground‑truth channels $\{{\mathbf{H}}_{\mathrm{v,aug},i}\}_{i=1}^{N_\mathrm{aug}}$ using an NMSE or similar distance measure.  
    Using NMSE can however be misleading. Picosecond-level timing errors introduce effectively random phase rotations that hardly affect beam patterns or link budgets, yet can inflate NMSE to its maximum. Matching absolute phase across positions is therefore not practical \cite{Hoydis2023SionnaRT}.  Metrics insensitive to a common phase rotation—such as beam-selection accuracy or downstream-task performance provide a more meaningful assessment of spatial consistency.
    

\section{Proposed Techniques for Synthetic Channel Generation}
\label{Proposed Methods}

In this section, we explain how the diffusion model can be used to generate synthetic channels from positional data. In Section \ref{lang_diff}, we discuss the concept of the score function, the key component of diffusion models, and how to train it. In Section \ref{cDDIM}, we provide an overview of Conditional DDIM (cDDIM), a specific diffusion-based generative model, including its training process and optimization algorithm. Next, we provide a theoretical analysis of diffusion-based generative models to justify our empirical experiments in Section \ref{Theoretical analysis}.

\subsection{Capturing Channel Distribution via Denoising Score Matching}
\label{lang_diff}

Suppose we want to generate channels according to the conditional channel distribution $p(\mathbf{H}_\mathrm{v}|\mathbf{x})$ for a given position $\mathbf{x}$. However, we do not have the distribution $p(\mathbf{H}_\mathrm{v}|\mathbf{x})$ but rather have just a collection of measurements $\{(\mathbf{x}_{\mathrm{train},i}, \mathbf{H}_{\mathrm{v,train},i})\}_{i=1}^{N_\mathrm{train}}$.
In the following, we (a) discuss the forward (noise-adding) vs.\ backward (denoising) processes (b) explain the score function and how it allows us to generate channels and (c) discuss how to learn the score function solely from channel samples.

\smallskip
\textbf{Forward vs.\ Backward Processes.}
In a typical diffusion setup, one defines a forward process that gradually adds noise to a clean channel sample $\mathbf{H}_\mathrm{v}[0]$ until it becomes nearly Gaussian at $\mathbf{H}_\mathrm{v}[T]$. Conversely, a backward process progressively denoises $\mathbf{H}_\mathrm{v}[T]$ step-by-step to recover a sample from the target distribution. In our formulation, Eq.~\eqref{dtLang} represents the discrete backward recursion, where $t$ decreases from $T$ down to $0$.

\smallskip \textbf{Score Function.} To generate synthetic channels that follow
$p(\mathbf{H}_\mathrm{v}|\mathbf{x}),$
we utilize the concept of a \emph{score function}, defined as
$\nabla_{\mathbf{H}_\mathrm{v}|\mathbf{x}} \log p(\mathbf{H}_\mathrm{v}|\mathbf{x}).$
This score function indicates the direction in which the log-density increases and is employed with Langevin dynamics \cite{song:19} to generate samples. In particular, for $i = 1, \ldots, N_t$, $j = 1, \ldots, N_r$, and for $t \in \{1, \ldots, T\}$, 
the discretized backward diffusion update is given by
\begin{align}
\label{dtLang}
\mathbf{H}_\mathrm{v}[t-1] = \frac{1}{\sqrt{1-\beta[t]}} \Bigl(\mathbf{H}_\mathrm{v}[t] + \beta[t]\nabla_{\mathbf{H}_\mathrm{v}[t]|\mathbf{x}} \log p\bigl(\mathbf{H}_\mathrm{v}[t]|\mathbf{x}\bigr)\Bigr),
\end{align}
where \(\beta[t]\) is a time-dependent parameter that controls the step size and the influence of the score function.
This update represents a discretized backward diffusion process, where the recursion is set so that when \(t=T\), \(\mathbf{H}_\mathrm{v}[T]\) is pure noise, and as \(t\) decreases to 0, \(\mathbf{H}_\mathrm{v}[0]\) converges to a sample from the desired distribution \(p(\mathbf{H}_\mathrm{v}|\mathbf{x})\). In fact, as \(T \to \infty\) and \(\sigma(t) \to 0\), the distribution of \(\mathbf{H}_\mathrm{v}[0]\) converges to the true density \cite{SDEBOOK}.

\smallskip \textbf{Learning Score Function from Samples: Denoising Score Matching. \ } 
Before we derive the denoising score matching procedure, we note that predicting noise in a partially noised sample is mathematically equivalent to estimating the score function. This equivalence arises from denoising score matching techniques \cite{Pas:11}.
The key idea is to parametrize a DNN, denoted as $\mathbf{S}(\mathbf{H}_\mathrm{v}|\mathbf{x},t;\mathbf{\Theta})$, where $\mathbf{\Theta}$ represents the parameters of the network, to approximate the score function and learn $\mathbf{\Theta}$ from samples. Here, $t$ represents the inference step, but in this section, we will ignore $t$ for simplicity and focus on the core concept of score function learning.
First, explicit score matching, which directly matches the DNN $\mathbf{S}(\mathbf{H}_\mathrm{v}|\mathbf{x};\mathbf{\Theta})$ with the score function $\nabla_{\mathbf{H}_\mathrm{v}|\mathbf{x}} \log p(\mathbf{H}_\mathrm{v}|\mathbf{x})$, can be written as minimizing the following loss function to train a model,
\begin{equation}
\label{loss_exp}
\mathcal{L}_\mathrm{exp}(\mathbf{H}_{\mathrm{v}}|\mathbf{x}; \mathbf{\Theta}) = \frac{1}{2} \mathbb{E}_{\mathbf{H}_\mathrm{v}} \left\| \mathbf{S}(\mathbf{H}_\mathrm{v}|\mathbf{x}; \mathbf{\Theta}) - \nabla_{\mathbf{H}_\mathrm{v}|\mathbf{x}} \log p(\mathbf{H}_\mathrm{v}|\mathbf{x}) \right\|_F^2.
\end{equation}
We aim to learn $\mathbf{\Theta}$ that minimizes the loss function above in \eqref{loss_exp}, but it cannot be directly calculated since the score function is unknown.

\textit{Denoising score matching} does not require underlying channel distribution $p(\mathbf{H}_\mathrm{v}|\mathbf{x})$ to learn the score function from samples of channel matrices. The key idea is to create a rescaled and perturbed version of $\mathbf{H}_\mathrm{v}$ via adding random Gaussian noise to the channel, denoted by $\widetilde{\mathbf{H}}_\mathrm{v}$, for which the score of the conditional distribution $p(\widetilde{\mathbf{H}}_\mathrm{v}|{\mathbf{H}}_\mathrm{v},\mathbf{x})$ can be easily computed analytically. Concretely, we define the perturbed channel as
\begin{equation}
\widetilde{\mathbf{H}}_{\mathrm{v}} = \alpha \mathbf{H}_{\mathrm{v}} + \sigma \mathbf{N},
\end{equation}
where $\alpha$ is a scaling constant, $\mathbf{N}_{ij} \sim \mathcal{CN}(0, 1)$ $\text{ for } i = 1, \ldots, N_t, j = 1, \ldots, N_r$, and 
$\sigma$ is the noise standard deviation (with a slight abuse of notation).

Then, we can define the conditional score function as $\nabla_{\sigma {\mathbf{N}}} \log p(\widetilde{\mathbf{H}}_\mathrm{v}|\mathbf{H}_\mathrm{v},\mathbf{x}) = \nabla_{\sigma \mathbf{N}} \log p(\sigma \mathbf{N}) = - \frac{\mathbf{N}}{\sigma}$, which is available to us since we generated $\mathbf{N}$. 
%
%
%
Next, we define the denoising score matching loss function as
\begin{equation}
\label{denoising_loss}
\begin{aligned}
&\mathcal{L}_\mathrm{den}(\widetilde{\mathbf{H}}_{\mathrm{v}}|\mathbf{H}_{\mathrm{v}},\mathbf{x}; \mathbf{\Theta}) \\
&= \frac{1}{2} \mathbb{E}_{{\mathbf{N}_{ij} \sim \mathcal{CN}(0, 1)}} \left\|{\mathbf{S}}(\widetilde{\mathbf{H}}_{\mathrm{v}}|\mathbf{x}; \mathbf{\Theta}) - \nabla_{\sigma \mathbf{N}} \log p(\widetilde{\mathbf{H}}_\mathrm{v}|\mathbf{H}_\mathrm{v},\mathbf{x}) \right\|_F^2 \\
&= \frac{1}{2} \mathbb{E}_{{\mathbf{N}_{ij} \sim \mathcal{CN}(0, 1)}} \frac{1}{\sigma} \left\| \widetilde{\mathbf{S}}(\widetilde{\mathbf{H}}_{\mathrm{v}}|\mathbf{x}; \mathbf{\Theta}) - \mathbf{N} \right\|_F^2, \\
\end{aligned}
\end{equation}
where $\widetilde{\mathbf{S}}(\cdot|\cdot; \mathbf{\Theta}) \equiv -\sigma \mathbf{S}(\cdot|\cdot; \mathbf{\Theta})$.
The above loss can be calculated without knowing the distribution $p(\mathbf{H}_\mathrm{v}|\mathbf{x})$, since $\nabla_{\sigma \mathbf{N}} \log p(\widetilde{\mathbf{H}}_\mathrm{v}|\mathbf{H}_\mathrm{v},\mathbf{x}) = - \frac{\mathbf{N}}{\sigma}$, which we know. 

The following proposition shows that explicit score matching and denoising score matching are equivalent.

\begin{proposition}\label{proposition:1}
(Adopted from \cite{Pas:11}): Assuming that $\log p(\widetilde{\mathbf{H}}_{\mathrm{v}}|\mathbf{H}_{\mathrm{v}},\mathbf{x})$ is differentiable with respect to $\widetilde{\mathbf{H}}_{\mathrm{v}}$, minimizing $\mathcal{L}_\mathrm{exp}(\mathbf{H}_{\mathrm{v}}|\mathbf{x}; \mathbf{\Theta})$ is equivalent to minimizing $\mathcal{L}_\mathrm{den}(\widetilde{\mathbf{H}}_{\mathrm{v}}|\mathbf{H}_{\mathrm{v}},\mathbf{x}; \mathbf{\Theta})$.
\end{proposition}
\begin{proof}
Follows Appendix in \cite{Pas:11}.
\end{proof}

Proposition \ref{proposition:1} means we can perform score matching without knowing the underlying distribution $p(\mathbf{H}_\mathrm{v}|\mathbf{x})$. By training the neural network  $\widetilde{\mathbf{S}}(\widetilde{\mathbf{H}}_\mathrm{v}|\mathbf{x};\mathbf{\Theta})$ to converge to a known $\mathbf{N}$ in a supervised fashion, we can effectively learn the score function $\nabla_{\mathbf{H}_\mathrm{v}|\mathbf{x}} \log p(\mathbf{H}_\mathrm{v}|\mathbf{x})$. 

Therefore, by leveraging denoising score matching, without making any assumptions about the underlying distribution $p(\mathbf{H}_\mathrm{v}|\mathbf{x})$ and using only $\{(\mathbf{x}_{\mathrm{train},i}, \mathbf{H}_{\mathrm{v,train},i})\}_{i=1}^{N_\mathrm{train}}$, we can train the score function and sample the channel from the trained model. The exact training and sampling algorithms are described in the next section.

\subsection{Conditional DDIM (cDDIM)}
\label{cDDIM}

To train the score function $\widetilde{\mathbf{S}}(\cdot|\cdot,\cdot; \mathbf{\Theta})$, we optimize the parameters $\mathbf{\Theta}$ as explained briefly in Section \ref{lang_diff}. 
This section focuses on explaining our method in detail, including the structure of our cDDIM model and its training and inference processes.

\label{U-net_model}
\begin{figure*}[!tbp]
  \centering
  \includegraphics[width=0.75\linewidth]{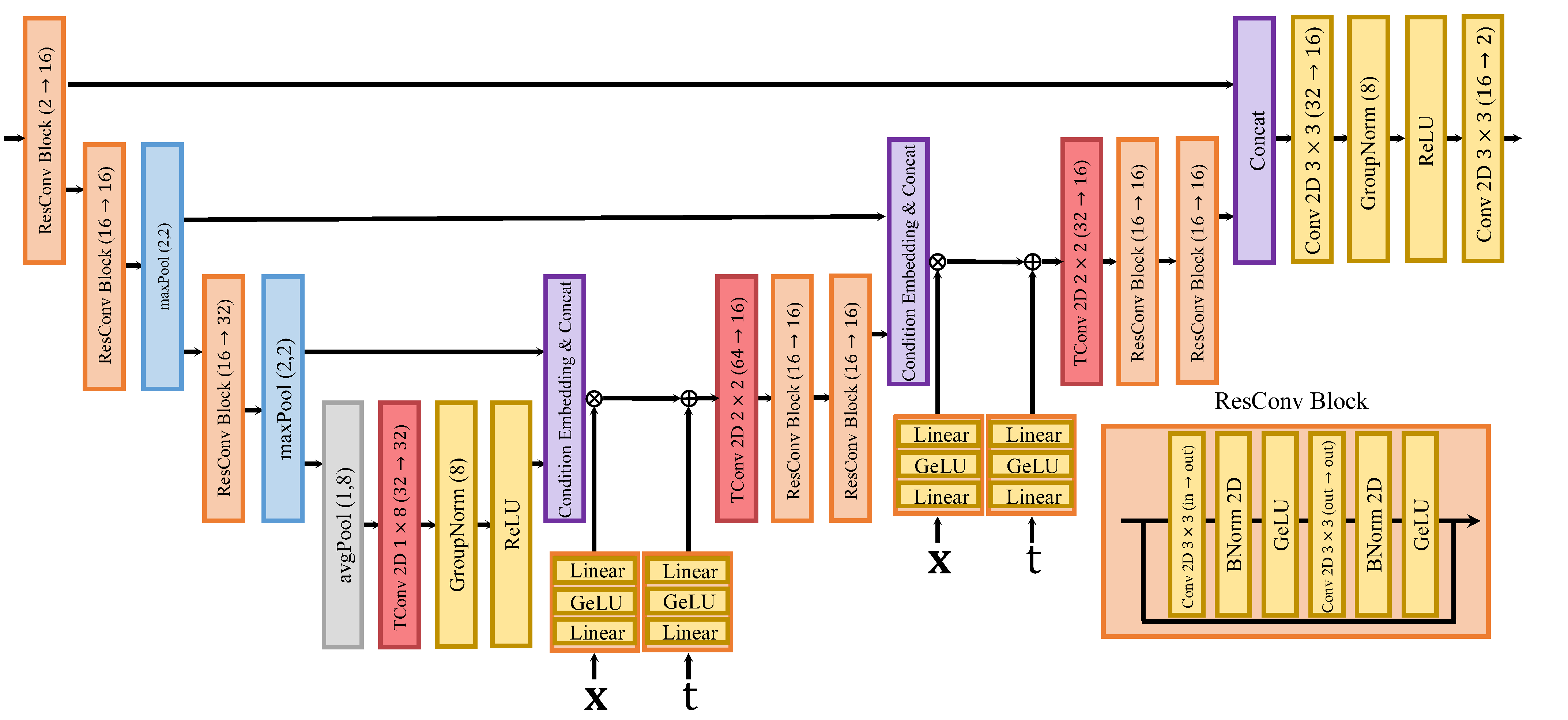}
  \caption{The cDDIM model architecture.}\label{figure:appendix.pdf}
\end{figure*}

\begin{figure}[!tbp]
  \centering
  \includegraphics[width=1.0\linewidth]{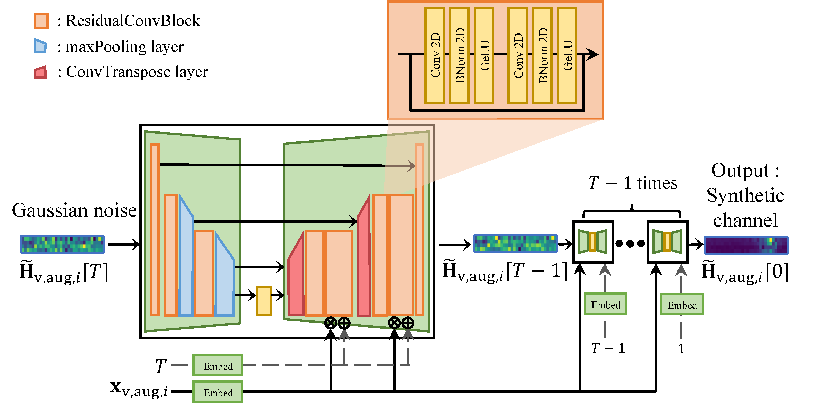}
  \caption{Inference process: Iteratively adding and denoising Gaussian noise over $T$ iterations to generate a synthetic channel from the conditional UE position input $\mathbf{x}_{\mathrm{v,aug,i}}$.}\label{figure:DDIM.pdf}
\end{figure}

\smallskip \textbf{Architecture of Our Method: The cDDIM Model.}
We implement our cDDIM model using a U‑net structure \cite{unet}. In addition, our U‑net architecture is relatively compact compared to modern computer vision models because it avoids self‑ and cross‑attention layers. This design reduces the overall parameter count and prevents overfitting on small datasets, while still capturing the essential features for channel generation.
The model takes the conditional input, UE position $\mathbf{x}$, and the inference step $t$, iterating T times from $t = T$ to $t = 1$. The conditional input $\mathbf{x}$ is embedded and elementwise multiplied with the concatenated vector, and the time step $t$ is elementwise added after embedding. The model structure is shown in Fig. \ref{figure:appendix.pdf}, and the entire process is illustrated in Fig. \ref{figure:DDIM.pdf}.

\smallskip \textbf{Training Process of the cDDIM Model.}
Diffusion-based generative models operate by learning a denoising process across various noise levels. The training process for a conditional DDIM (cDDIM) is described in Algorithm \ref{alg:train_classifier_free}.

\begin{algorithm}
\caption{Training cDDIM}
\label{alg:train_classifier_free}
\begin{algorithmic}[1]
\Require{Precomputed schedules $\{\overline{\alpha}[t]\}_{t=1}^{T}$, a model $\widetilde{\mathbf{S}}(\cdot \mid \cdot,\cdot;\cdot)$}
\Input{Channel matrices $\{\mathbf{H}_{\mathrm{v,train},i}\}_{i=1}^{N_\mathrm{train}}$ , corresponding UE positions $\{\mathbf{x}_{\mathrm{train},i}\}_{i=1}^{N_\mathrm{train}}$, initial model parameters $\mathbf{\Theta}$}
\Repeat
    \For{$i = 1$ to $N_{\mathrm{train}}$}
        \State $\mathbf{H}_{\mathrm{v,train},i}[0] \leftarrow \mathbf{H}_{\mathrm{v,train},i}$ 
        \State $t \sim \text{Uniform}(1, T)$ 
        \State $\mathbf{N}_{ab}[t] \sim \text{i.i.d.}, \mathcal{CN}(0, 1) \text{ for } \forall a, b$
        \State $\mathbf{H}_{\mathrm{v,train},i}[t] \leftarrow \sqrt{\overline{\alpha}[t]} \mathbf{H}_{\mathrm{v,train},i}[0] + \sqrt{1 - \overline{\alpha}[t]} \mathbf{N}[t]$ 
        \State $\mathbf{\Theta} \leftarrow \mathbf{\Theta} - \eta \nabla_{\mathbf{\Theta}} \| \widetilde{\mathbf{S}}(\mathbf{H}_{\mathrm{v,train},i}[t]|\mathbf{x}_{\mathrm{train},i},t;\mathbf{\Theta}) - \mathbf{N}[t] \|_F^2$
    \EndFor
\Until{converged}
\Output{Trained model $\widetilde{\mathbf{S}}(\cdot|\cdot,\cdot;\mathbf{\Theta})$}
\end{algorithmic}
\end{algorithm}

As shown in line 1 of Algorithm \ref{alg:train_classifier_free}, our input to the channel consists of channel matrices and corresponding UE position pairs. 
In this context, $\overline{\alpha}[t]$ represents the cumulative product of a predefined scaling schedule $\alpha$ over time steps, defined as 
\[
\overline{\alpha}[t] = \prod_{u=1}^t \alpha[u] \quad \text{with} \quad \alpha[t] = 1 - \beta[t],
\] where \(\beta[t]\) is described in Section~\ref{lang_diff} \cite{DDPM}.
The variable $\sigma[t]$, defined as $\sigma[t] = \sqrt{1 - \overline{\alpha}[t]}$, governs the noise level at each step and is used in the denoising process.
Although it is a sequential discrete process, we can represent $\mathbf{H}_{\mathrm{v,train},i}[t]$ in terms of $\mathbf{H}_{\mathrm{v,train},i}[0]$ and noise $\mathbf{N}[t]$, which is 
\begin{equation}
\label{dt_lang}
\begin{split}
&\mathbf{H}_{\mathrm{v,train},i}[t] = \sqrt{\overline{\alpha}[t]} \mathbf{H}_{\mathrm{v,train},i}[0] + \sqrt{1 - \overline{\alpha}[t]} \mathbf{N}[t].
\end{split}
\end{equation}

Then, we update $\mathbf{\Theta}$ by calculating the gradient of the difference between $\widetilde{\mathbf{S}}$ with the input of the noise-added channel $\mathbf{H}_{\mathrm{v,train},i}[t]$ and the noise $\mathbf{N}[t]$, as shown in line 8. The output of Algorithm \ref{alg:train_classifier_free} is the trained model $\widetilde{\mathbf{S}}(\cdot|\cdot,\cdot;\mathbf{\Theta})$, which is used for the inference process, as explained next.

\smallskip \textbf{Inference Process of the cDDIM Model.} Our goal is to generate $\mathbf{H}_\mathrm{v} \sim p(\mathbf{H}_\mathrm{v}|\mathbf{x})$ for a given UE position $\mathbf{x}$ as input. 
The sampling process is described in detail in Algorithm \ref{alg:sample_classifier_free}.

\begin{algorithm}
\caption{Sampling from a trained cDDIM}
\label{alg:sample_classifier_free}
\begin{algorithmic}[1]
\Require{Precomputed schedules $\{\overline{\alpha}[t]\}_{t=1}^{T}$, pretrained model $\widetilde{\mathbf{S}}(\cdot|\cdot,\cdot;\mathbf{\Theta})$}
\Input{UE positions $\{\mathbf{x}_{\mathrm{aug},i}\}_{i=1}^{N_\mathrm{aug}}$}
\State $\widetilde{\mathbf{H}}_{\mathrm{v,aug},i,ab}[T] \sim \mathcal{CN}(0, 1) \text{ for } \forall a, b$
\For{$i = 1$ to $N_{\mathrm{aug}}$}
    \For{$t = T, \ldots, 1$}
        \State \small $\widetilde{\mathbf{H}}_{\mathrm{v,aug},i}[t-1]\!\leftarrow\!\sqrt{1-\overline{\alpha}[t\!-\!1]}\,\widetilde{\mathbf{S}}\bigl(\widetilde{\mathbf{H}}_{\mathrm{v,aug},i}[t]\!\!\mid\!\mathbf{x}_{\mathrm{aug},i},t;\!\mathbf{\Theta}\bigr)$
        \Statex \hspace{0.2em} $+\sqrt{\overline{\alpha}[t-1]} \left(\frac{\widetilde{\mathbf{H}}_{\mathrm{v,aug},i}[t]-\sqrt{1-\overline{\alpha}[t]} \ \widetilde{\mathbf{S}}(\widetilde{\mathbf{H}}_{\mathrm{v,aug},i}[t]|\mathbf{x}_{\mathrm{aug},i},t;\mathbf{\Theta})}
        {\sqrt{\overline{\alpha}[t]}}\right)$
    \EndFor
\EndFor
\Output 
$\{\widetilde{\mathbf{H}}_{\mathrm{v,aug},i}[0]\}_{i=1}^{N_\mathrm{aug}}$
\end{algorithmic}
\end{algorithm}

When sampling from the model, we need to perform a backward process.
The backward process transforms arbitrary Gaussian noise into a clean image through a sequence of $T$ denoising steps.

Since we are trying to train the deterministic function between $\mathbf{x}$ and $\mathbf{H}_\mathrm{v}[0]$, we use DDIM \cite{DDIM}, which follows a deterministic generation process.
While we focus on DDIM due to its deterministic sampling,
we note that the DDPM-based augmentation yields comparable
results for channel distribution modeling. In
practice, choosing between DDPM and DDIM has negligible impact on
overall performance, but we opt for DDIM's deterministic property
because it yields more consistent, position-specific channel realizations
without introducing additional random variation.
While the previous equation \eqref{dtLang} utilized the score function, the current approach approximates this process using $\widetilde{\mathbf{S}}$ in the DDIM framework. The DDIM sampling equation is
\begin{equation}
\label{cDDIM:inf}
\begin{split} 
&\widetilde{\mathbf{H}}_{\mathrm{v,aug},i}[t-1] \!=\! \sqrt{1-\overline{\alpha}[t-1]} \ \widetilde{\mathbf{S}}(\widetilde{\mathbf{H}}_{\mathrm{v,aug},i}[t]|\mathbf{x}_{\rm{aug},i},t;\mathbf{\Theta}) \\
&+\!\sqrt{\overline{\alpha}[t\!-\!1]} \!\left(\!\frac{\widetilde{\mathbf{H}}_{\mathrm{v,aug},i}[t]\!-\!\sqrt{1\!-\!\overline{\alpha}[t]} \ \widetilde{\mathbf{S}}(\widetilde{\mathbf{H}}_{\mathrm{v,aug},i}\![t]\!|\mathbf{x}_{\rm{aug},i},t;\mathbf{\Theta})}
{\sqrt{\overline{\alpha}[t]}}\!\right)\!.
\end{split}
\end{equation}
Samples are generated from latent variables using a fixed procedure, without any stochastic noise involved in \eqref{cDDIM:inf}. Consequently, the model functions as an implicit probabilistic model. This process is repeated for all $N_\text{aug}$ UE positions, as shown in lines 3 to 8 of Algorithm \ref{alg:sample_classifier_free}. The final output of the model is an augmented dataset $\{\widetilde{\mathbf{H}}_{\mathrm{v,aug},i}[0]\}_{i=1}^{N_\mathrm{aug}}$, which we use for downstream wireless communication tasks.




\subsection{Theoretical Analysis}
\label{Theoretical analysis}

In this section, we provide theoretical analysis to answer the following question: Can a diffusion model trained with only $N_\mathrm{train}$ samples reliably learn the score function, with theoretical guarantees? (We exclude the conditioning of $\mathbf{H}_\mathrm{v}$ on $\mathbf{x}$ and $t$ in this section for notational convenience.)

Diffusion models are a recent development, and their analysis is well-understood only in certain special cases, such as Gaussian data \cite{Bruno2023}, which can also be thought of as Rayleigh fading channels. Nevertheless, based on recent findings, we provide theoretical guarantees on the convergence of these models in terms of the latent dimension of the MIMO channels, which we define formally in Remark \ref{remark:2}. The key insight is as follows: 
leveraging the fact that sparse MIMO channels have low rank $r$, we demonstrate that the crucial factor for the convergence of the diffusion model is not the dimension of the channel itself, $N_t \times N_r$, but the dimension of the underlying latent vector, $d \leq r$. 

Before we present Remark 1, we note that the remark is established for the continuous Langevin process $\mathbf{H}_\mathrm{v}(t)$ which is related to the discrete Langevin process $\mathbf{H}_\mathrm{v}[t]$. We will not detail this relationship in this context \cite{minshuo23}. \(\mathcal{O}(\cdot)\) describes the growth rate of a function as the input size increases. For example, \(\mathcal{O}(n^2)\) means the function grows quadratically. Tilde notation \(\widetilde{\mathcal{O}}(n)\) is similar but includes slower-growing factors like logarithms.

\begin{remark}
\label{remark:2}
Consider the continuous Langevin process $\mathbf{H}_{\mathrm{v}}(t)$, where $\mathbf{H}_{\mathrm{v}}(0)$ represents the pure channel matrix and $\mathbf{H}_{\mathrm{v}}(T)$ represents Gaussian noise. If the channel distribution can be expressed as \(\underline{\mathbf{H}_{\mathrm{v}}}(0) = \mathbf{A} \mathbf{z}(\mathbf{x})\), where \(\underline{\mathbf{H}_{\mathrm{v}}}(0) \in \mathbb{C}^{N_t N_r}\), \(\mathbf{A} \in \mathbb{C}^{N_t N_r \times d}\) with orthonormal columns, and \(\mathbf{z}(\mathbf{x}) \in \mathbb{C}^{d}\) is a low-dimensional function vector of $\mathbf{x}$, then from Theorem 1 in \cite{minshuo23}, the difference in the score function, $\mathbb{E}_{0 \leq t \leq T}\left\|{\mathbf{S}}\left(\mathbf{H}_{\mathrm{v}}(t); \mathbf{\Theta}\right) - \nabla_{\mathbf{H}_\mathrm{v}(t)} \log p(\mathbf{H}_\mathrm{v}(t))\right\|$, can be bounded as $\widetilde{\mathcal{O}}\left(N_\text{train}^{-\frac{2}{d+6}}\right).$ 
\end{remark}

Remark \ref{remark:2} tells us that if we can assume the channel matrix distribution can be projected onto a low-rank space, we can establish an asymptotic error bound for the diffusion model when using a finite amount of data. Given that our dataset consists of sparse MIMO channels in beamspace, this assumption is highly plausible. Thus, even for high-dimensional channel data, the low-rank and sparse nature of mmWave channels allows the diffusion model to learn the score function with only a finite amount of data, as it is the dimension of the underlying latent vector that is crucial.

However, there are limitations to the above analysis. First, we generally assumed that \(\mathbf{H}_{\mathrm{v}}\) can be expressed in the form \(\mathbf{A}\mathbf{z}(\mathbf{x})\), which is not the case for channels derived from site-specific simulators. If we force this form, the dimension \(d\) of \(\mathbf{z}\) would become very large. Second, for example, even if \(d = 6\), which is reasonably small, the bound on the score function function \(\mathbf{S}\) is still \(\widetilde{\mathcal{O}}(N_\text{train}^{-\frac{1}{6}})\). 

This means that while increasing $N_\text{train}$ does reduce the error, to reduce the error by 10 times, we would theoretically need $10^6$ times more data, which is quite impractical. Nevertheless, in practice, we observe that the diffusion model converges without requiring such an enormous amount of data, suggesting that the theoretical bounds may be conservative and that the model’s practical performance improves with a moderate $N_\text{train}$ size.

\subsection{Complexity Analysis}

Table~\ref{tab:time_comparison} compares both the parameter count and time complexity of our cDDIM-based approach with three baselines: adding Gaussian noise, ChannelGAN~\cite{Xia:22}, and the one-shot U-net. Let $E$ be the number of training epochs, $N_{\mathrm{train}}$ the training set size, $T$ the number of diffusion steps, and $U$ the cost of a single forward--backward pass in the U-net. For GAN-based augmentation, $G_{\mathrm{fw}}$/$G_{\mathrm{bw}}$ and $D_{\mathrm{fw}}$/$D_{\mathrm{bw}}$ denote the forward/backward pass costs for the generator and discriminator, respectively.

While the one‑shot U‑net shares the same training time as the diffusion model, it requires only a single step during inference, making it $T$ times faster.
Although adding Gaussian noise, ChannelGAN, and the one‑shot U‑net are generally faster to train or deploy, they yield significantly weaker performance—especially on downstream or out-of-distribution tasks, as we will show later. By contrast, cDDIM offers markedly higher realism and task accuracy at an additional cost of $\mathcal{O}(T \times U)$ per sample during inference. 

In terms of parameters, cDDIM has about 6 times more than ChannelGAN (i.e., $\sim\!62\text{k}$ vs.\ $\sim\!10\text{k}$), but increasing ChannelGAN’s dimensionality to match cDDIM’s parameter count does not improve ChannelGAN’s performance. One-shot U-net shares cDDIM’s backbone (both ~62k), and though they share the same training complexity, one-shot U-net requires only one inference step, making it $T$ times faster than cDDIM. While adding Gaussian noise, ChannelGAN, and one-shot U-net are faster, they yield weaker performance—especially out-of-distribution. By contrast, cDDIM offers markedly higher realism and task accuracy at an extra inference cost of $\mathcal{O}(T \times U)$ per sample. 
Since channel generation is typically performed \emph{offline}, this additional overhead is acceptable compared to the substantial performance benefits.


\begin{table}[h]
\centering
\scriptsize
\setlength{\tabcolsep}{4pt}
\caption{Time complexity and parameter size comparison.}
\label{tab:time_comparison}
\resizebox{\columnwidth}{!}{
\begin{tabular}{lccc}
\toprule
\textbf{Method} & \textbf{\#Parameters} & \textbf{Training Complexity} & \textbf{Inference Complexity} \\
\midrule
\textbf{Add Gaussian noise} 
& 0  
& $\mathcal{O}(N_\mathrm{train})$ 
& $\mathcal{O}(N_\mathrm{aug})$ \\
\textbf{ChannelGAN} 
& 10k 
& $\mathcal{O}\bigl(E \times N_\mathrm{train} \times (G_{\mathrm{fw}} + G_{\mathrm{bw}} + D_{\mathrm{fw}} + D_{\mathrm{bw}})\bigr)$ 
& $\mathcal{O}(N_\mathrm{aug} \times G_{\mathrm{fw}})$ \\
\textbf{One-shot U-net} 
& 62k 
& $\mathcal{O}(E \times N_\mathrm{train} \times U)$ 
& $\mathcal{O}(N_\mathrm{aug} \times U)$ \\
\textbf{Diffusion (cDDIM)} 
& 62k 
& $\mathcal{O}(E \times N_\mathrm{train} \times U)$ 
& $\mathcal{O}(T \times N_\mathrm{aug} \times U)$ \\
\bottomrule
\end{tabular}
}
\end{table}

\section{Visualization and Evaluation}
\label{Visualization and Evaluation}
In this section, we evaluate the proposed dataset augmentation method through both qualitative visualization and quantitative analysis. We conduct two experiments:

\textbf{Experiment In‑Distribution (Exp ID):} Both the training and test datasets are collected from a 100\,m radius centered at the base station (BS).

\textbf{Experiment Out‑of‑Distribution (Exp OOD):} The training set is collected from a 100\,m radius, while the test set is collected from a donut-shaped region spanning 100\,m to 200\,m, to evaluate generalization under mismatched train/test distributions.

\begin{figure}[!tbp]
    \centering
    \includegraphics[width=0.6\columnwidth]{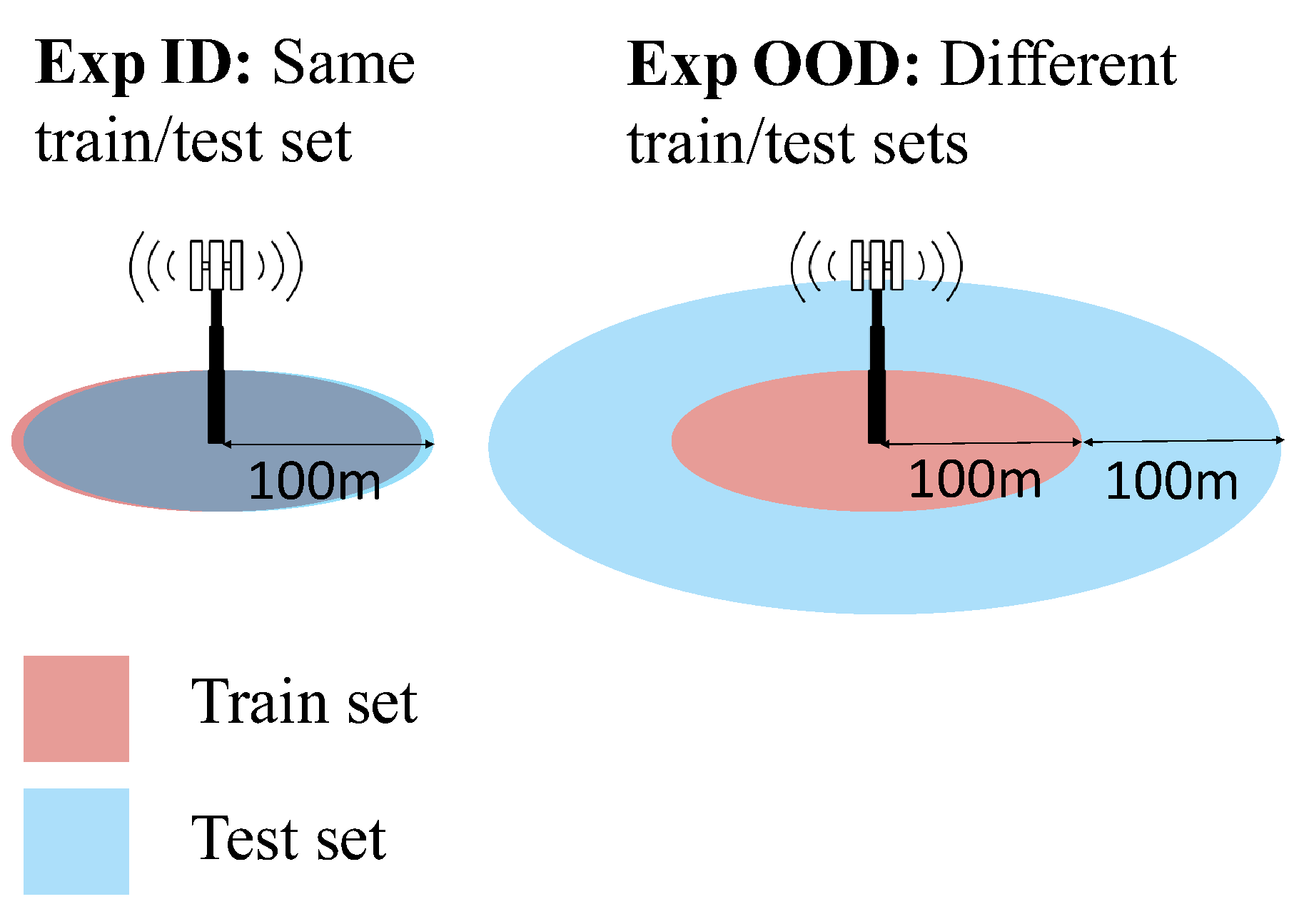}
    \caption{Illustration of the two experiment setups. 
    In Exp ID, both the training and test datasets are 
    within a 100\,m radius. In Exp OOD, the test dataset is 
    collected from a donut-shaped region spanning 100\,m to 200\,m 
    from the base station.}
    \label{fig:twoexp}
\end{figure}

Fig.~\ref{fig:twoexp} visualizes the two experiment settings. Section \ref{Simulation Setup} describes the simulation environment. In Section \ref{Visualizations}, the generated channel is visualized, demonstrating that our method produces accurate estimates when UE positions are given as conditional input.
Finally, Section~\ref{Beam selection} quantitatively analyzes the index of the maximum beam and shows that our cDDIM method provides a good approximation of the ground truth channel—even when there are no nearby users in the training dataset.

\subsection{Simulation Setup}
\label{Simulation Setup}

\begin{table}[ht]
\small
\centering
\caption{Simulation parameters of the proposed approach}\label{table:quadriga}
\begin{tabular}{| c | c |}
\hline
Simulation Environment & QuaDRiGa \\ \hline
Scenario Name & Berlin urban macro LOS \\ \hline
UE Pre-Augmentation $N_\text{train}$ & 100 -- 10,000 \\ \hline
UE Augmented Samples $N_\text{aug}$ & 90,000 \\ \hline
UE Inference Samples $N_\text{test}$ & 10,000 \\ \hline
BS Antenna $N_t$ & 32, $32 \times 1$ ULA \\ \hline
UE Antenna $N_r$ & 4, $4 \times 1$ ULA \\ \hline
DDIM Train Epochs & 50,000 \\ \hline
DDIM Sampling Steps & 256 \\ \hline
Carrier Frequency & 28 GHz \\ \hline
Bandwidth ($B$) & 20 MHz \\ \hline
UE Range & \begin{tabular}[c]{@{}l@{}} Train: 100\,m radius \\ Test\,(Exp~1): 0--100\,m \\ Test\,(Exp~2): 100--200\,m 
\end{tabular} \\ \hline
\end{tabular}
\end{table}

\smallskip \textbf{Channel Matrix Generation and Dataset Description.} 
We follow these steps in our simulation setup:

\begin{enumerate}
    \item We randomly locate $N_{\text{train}}$ training samples and determine user positions for $N_{\text{aug}}$ augmented samples in the Berlin urban macro LOS scenario using the QuaDRiGa simulator \cite{quadriga2023}.
    \item We generate an initial set of channels for the $N_{\text{train}}$ training samples and then generate $N_{\text{aug}}$ augmented channels by conditioning on the corresponding user locations.
    \item For evaluation, we compare the $N_{\text{aug}}$ augmented samples with an equally sized set of reference samples. Specifically, we visualize the generated channels in Fig.~\ref{figure:comparison_visualization.pdf} and compare peak index match probabilities in Fig.~\ref{figure:probabilitycomparison.pdf}.
\end{enumerate}

Although $N_{\text{train}}$ and $N_{\text{aug}}$ vary for the downstream tasks, in these experiments we set $N_{\text{train}} = 100$ and $N_{\text{aug}} = 10,000$.
The underlying channel generation follows the QuaDRiGa simulator, which produces realistic radio channel impulse responses for system-level simulations of mobile radio networks. Each DDIM inference takes 256 steps.

We apply min–max normalization to each subcarrier, dividing by its largest amplitude so that values range from 0 to 1 and ensuring uniform channel scaling for all experiments and downstream tasks.
Table \ref{table:quadriga} details the parameters and settings.

\begin{figure*}[!tbp]
  \centering
  \begin{subfigure}{0.49\linewidth}
    \centering
    \includegraphics[width=\linewidth, trim={50pt 10pt 100pt 30pt},clip]{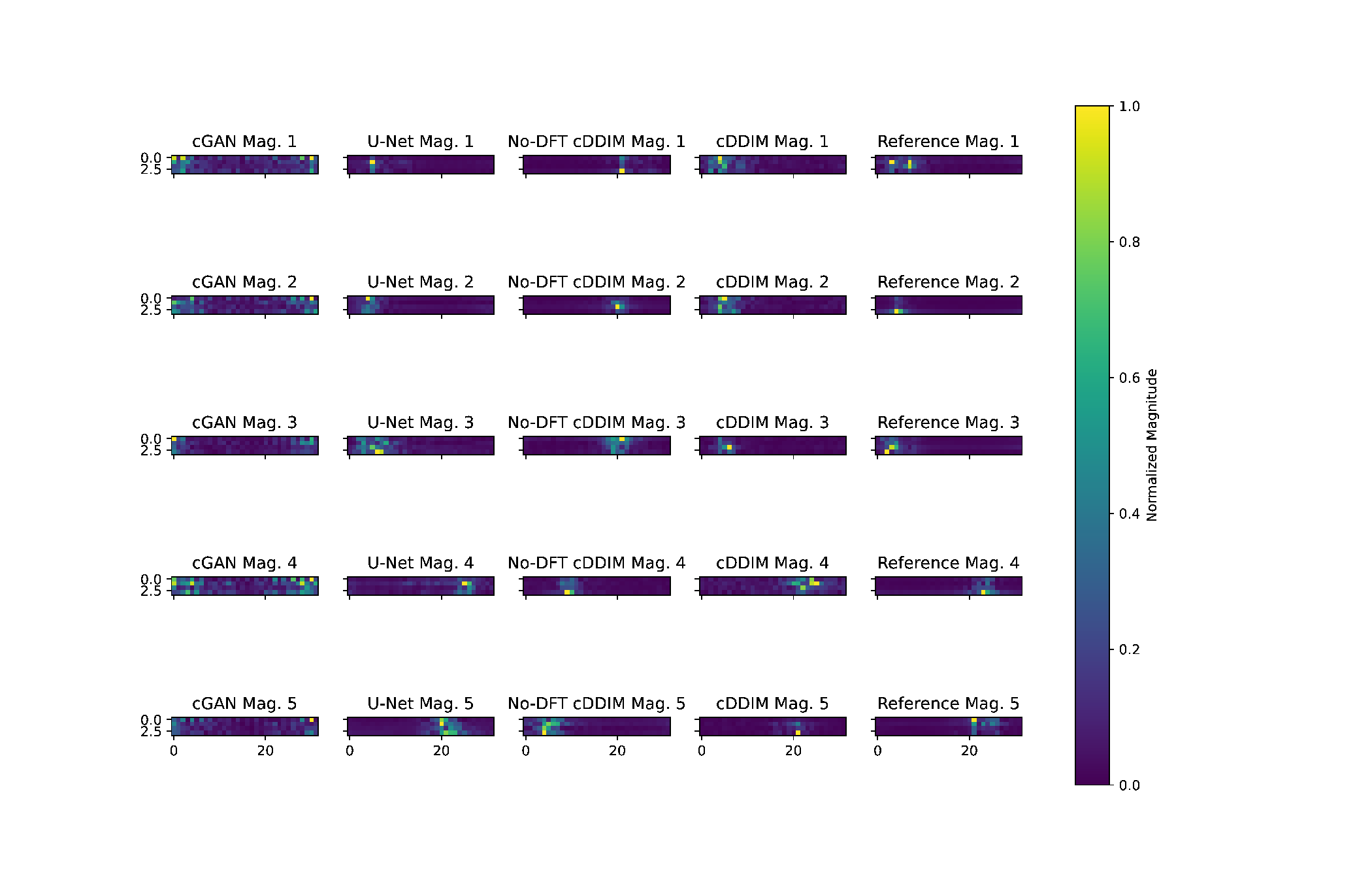}
    \caption{\textbf{Exp ID}: The training and test datasets are drawn from the same distribution.}
    \label{fig:exp1}
  \end{subfigure}
  \begin{subfigure}{0.49\linewidth}
    \centering
    \includegraphics[width=\linewidth, trim={50pt 10pt 100pt 30pt},clip]{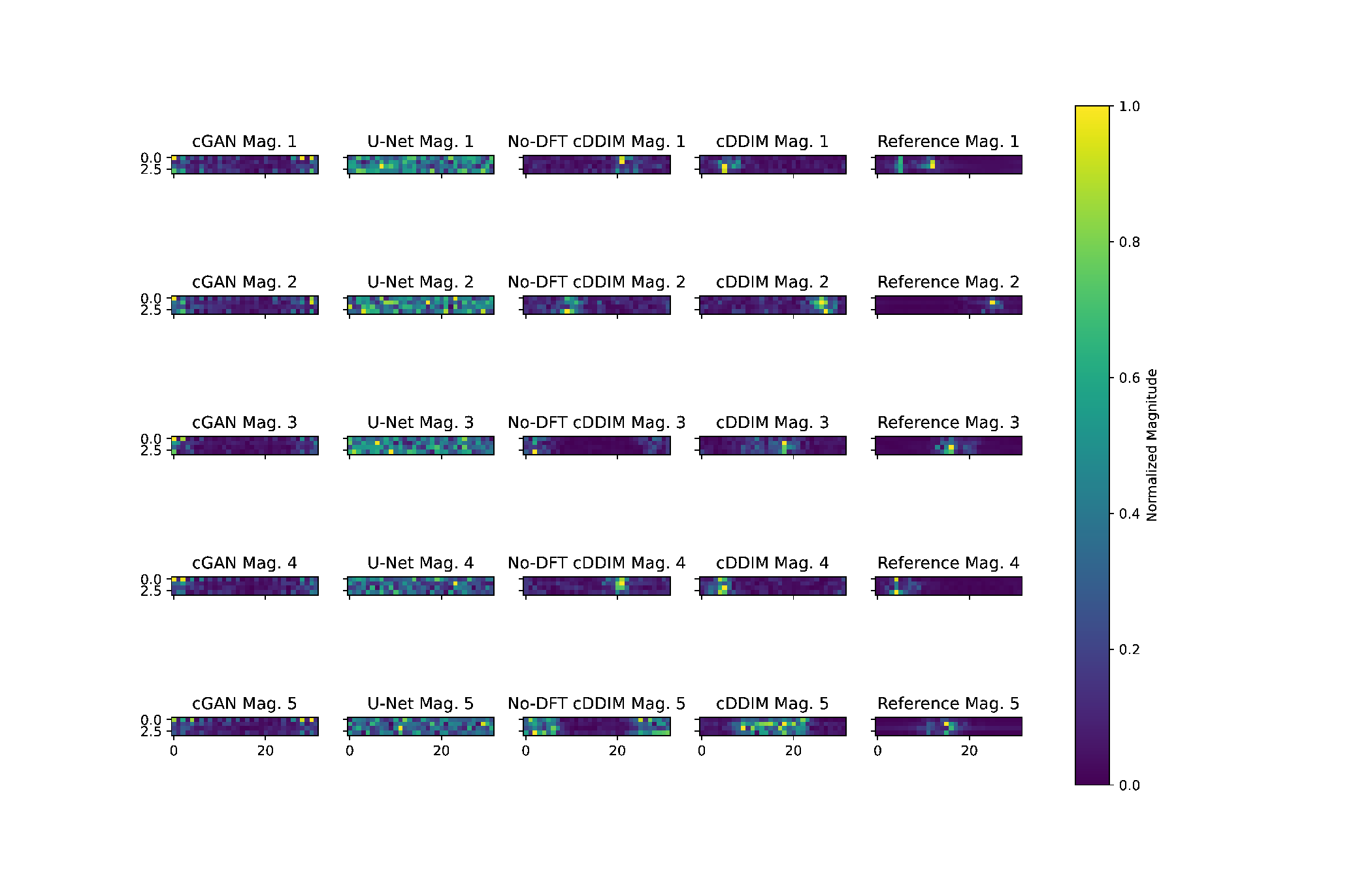}
    \caption{\textbf{Exp OOD}: The training and test datasets are drawn from different distributions.}
    \label{fig:exp2}
  \end{subfigure}
  \caption{Visualization of the magnitude of five randomly selected synthetic/reference channel examples. From left to right, each column shows channel samples generated by cGAN, one-shot U-net, cDDIM trained in the spatial domain, cDDIM, and the reference channel.}
  \label{figure:comparison_visualization.pdf}
\end{figure*}

\subsection{Visualizations of the Generated Channel}
\label{Visualizations}

\smallskip {\bf Observations from Generated Channel.} In our scenario, with a LOS path and using ULA antennas, the beamspace domain of the channel typically shows one main cluster with a significantly higher magnitude value than any other point. We define this as the peak, specifically examining the peak BS side index. In Fig. \ref{figure:comparison_visualization.pdf}, we compare five random channel samples generated by cGAN, cDDIM, one-shot U-net, and the reference channels at the same position. The visualization highlights how each method predicts the peak BS index in the LOS path.

Random UE antenna orientations \cite{3gpp} make the UE side order unpredictable. Therefore, we can predict the BS index but not the UE index. What we want to capture from the channel matrix is the index of the largest peak. 

\textbf{Baselines.} We compare with several baseline methods, either existing or newly developed by us.

\emph{No-DFT cDDIM:}
To emphasize the effect of beamspace transforms, we also consider a cDDIM variant trained directly in the \emph{spatial} domain without applying DFT. That is, this baseline attempts to learn the distribution of $\mathbf{H}$ itself, rather than its beamspace representation $\mathbf{H}_{\mathrm{v}}$. As shown below, this approach fails to replicate the distinctive LoS peak and multipath structure, which underscores why beamspace domain input is crucial.
 
\emph{cGAN (Conditional GAN):} For comparison, we use the cDDIM method and the conditional GAN (cGAN) method, which is a variant derived from ChannelGAN \cite{Xia:22}. ChannelGAN does not include positional data, so we implemented similar conditioning with our cDDIM method to ChannelGAN and named it cGAN. Our goal is to determine if it can learn the function that maps position to channel matrix using cGAN.

\emph{One‑shot U-net:}
Rather than iteratively denoising noisy channel samples (like cDDIM), this model uses a single forward pass of the same U-net backbone to map a random noise sample (plus the UE position) directly to a channel. Algorithm~\ref{alg:train_consistency_model} outlines the training procedure, where we reuse the cDDIM structure to train the U-net to output a clean channel in one step. Then, at inference, as shown in Algorithm~\ref{alg:sample_ct}, we simply pass pure noise through the model to generate a channel. Although this approach provides much lower inference latency, it may yield less accurate or less diverse samples than multi-step diffusion.

\begin{algorithm}
\caption{Training one-shot U-net}
\label{alg:train_consistency_model}
\begin{algorithmic}[1]
\Require{Precomputed schedules $\{\overline{\alpha}[t]\}_{t=1}^{T}$}
\Input{Channel matrices $\{\mathbf{H}_{\mathrm{v,train},i}\}_{i=1}^{N_\mathrm{train}}$ , corresponding UE positions $\{\mathbf{x}_{\mathrm{train},i}\}_{i=1}^{N_\mathrm{train}}$, initial model parameter $\mathbf{\Theta}$}
\Repeat
    \For{$i = 1$ to $N_{\mathrm{train}}$}
        \State $\mathbf{N}_{ab} \sim \text{i.i.d.}, \mathcal{CN}(0, 1) \text{ for } \forall a, b$ 
        \State { $\widehat{\mathbf{H}}_{\mathrm{v,train},i}\gets \widetilde{\mathbf{S}}(\mathbf{N}\,|\,\mathbf{x}_{\mathrm{train},i},\, t;\,\mathbf{\Theta})$
        }
        \State $\mathbf{\Theta} \leftarrow \mathbf{\Theta} - \eta \nabla_{\mathbf{\Theta}} \| \widehat{\mathbf{H}}_{\mathrm{v,train},i} - \mathbf{H}_{\mathrm{v,train},i}\|_F^2$
    \EndFor
\Until{converged}
\Output{Trained model $\widetilde{\mathbf{S}}(\cdot|\cdot,\cdot;\mathbf{\Theta})$}
\end{algorithmic}
\end{algorithm}

\begin{algorithm}
\caption{Sampling from one-shot U-net}
\label{alg:sample_ct}
\begin{algorithmic}[1]
\Require{pretrained model $\widetilde{\mathbf{S}}(\cdot|\cdot,\cdot;\mathbf{\Theta})$}
\Input{UE positions $\{\mathbf{x}_{\mathrm{aug},i}\}_{i=1}^{N_\mathrm{aug}}$}
\State $\mathbf{N}_{ab} \sim {\mathcal{CN}(0, 1)} \text{ for } \forall a, b$
\For{$i = 1$ to $N_{\mathrm{aug}}$}
    \State $\widetilde{\mathbf{H}}_{\mathrm{v,aug},i}\!\leftarrow\widetilde{\mathbf{S}}\bigl(\mathbf{N}\!\!\mid\!\mathbf{x}_{\mathrm{aug},i},t;\!\mathbf{\Theta}\bigr)$
\EndFor
\Output 
$\{\widetilde{\mathbf{H}}_{\mathrm{v,aug},i}\}_{i=1}^{N_\mathrm{aug}}$
\end{algorithmic}
\end{algorithm}

Fig.~\ref{figure:comparison_visualization.pdf} shows five randomly selected test samples for \textbf{Exp ID} and \textbf{Exp OOD}, respectively. From left to right, the columns display cGAN, one‑shot U‑net, cDDIM trained in the spatial domain, cDDIM (beamspace), and the ground‑truth (reference) channel.
We observe that cGAN produces channels lacking diversity, and consistently place peaks at similar coordinates in the synthetic channels, even though the reference channels have peaks at different coordinates.
In contrast, examining the fourth column (cDDIM, beamspace) shows that the BS-side peak index reliably aligns with that of the reference channel. This suggests that cDDIM can make accurate estimates given the UE coordinates, resulting in a dataset with correct predictions.
Also, training in the spatial
domain fails to reproduce the relationship between LoS peak and the
concentrated multipath structure—both of which are crucial for many
applications like beam alignment and channel compression. This result
is consistent with existing literature (e.g., \cite{CsiNet,doshi2022over})
showing that beamspace transforms help isolate significant paths and
reduce complexity. Indeed, the fact that only the beamspace version of cDDIM effectively learns the distribution underscores the inherent difficulty of modeling arbitrary distributions directly in the spatial domain.

A standout observation in Fig.~\ref{fig:exp2} is that the one‑shot U‑net fails to learn the channel structure, whereas cDDIM continues to capture it effectively. We will provide a more detailed analysis of this phenomenon in Section~\ref{Beam selection}.

\subsection{Quantitative Analysis}
\label{Beam selection}

To quantitatively evaluate the quality of the generated channels, we analyze the Line-of-Sight (LOS) peaks by examining the BS side index and comparing the differences between the peaks generated from conditional models and those from reference channels, as discussed qualitatively in the previous section. 


\begin{figure}[!tbp]
  \centering
  \includegraphics[width=0.9\linewidth, viewport=50 180 550 580, clip]{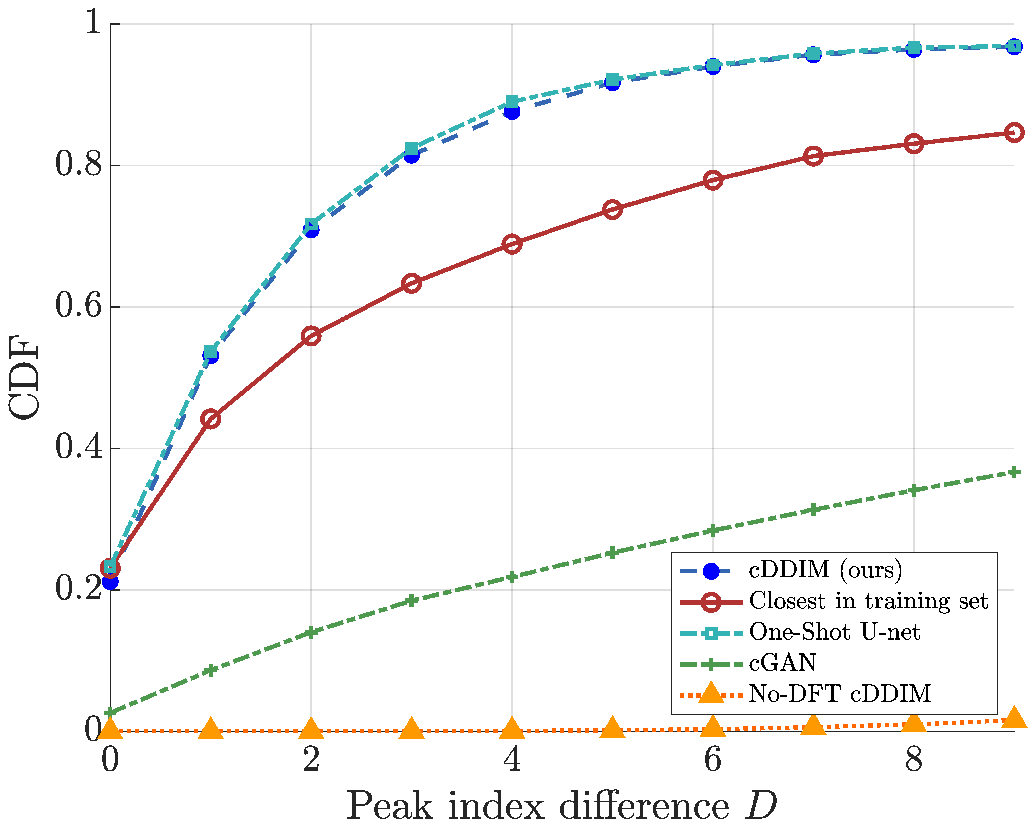}
  \caption{Peak index match probabilities versus $D$ across various channel inference techniques based on positional information for \textbf{Exp ID}. The results show that both the cDDIM method and the one-shot U-net achieve the highest match probabilities, indicating their superior performance.}\label{figure:probabilitycomparison.pdf}
\end{figure}

\begin{figure}[!tbp]
  \centering
  \includegraphics[width=0.9\linewidth, trim={70pt 190pt 60pt 200pt}]{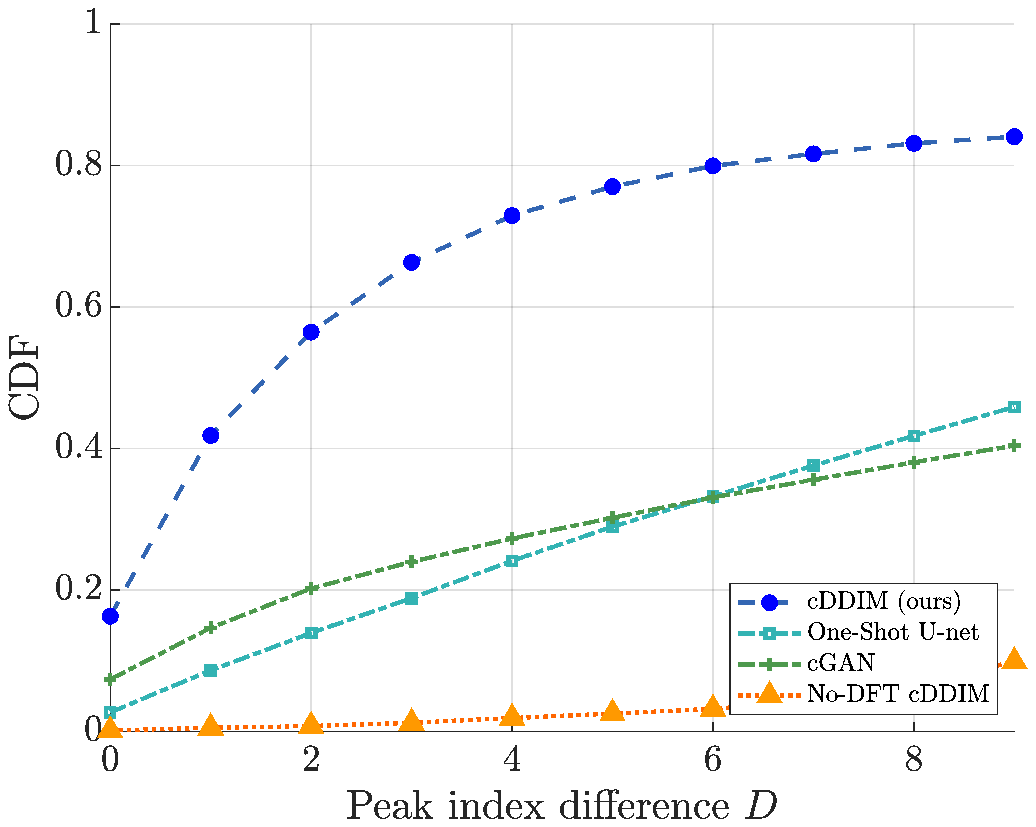}
  \caption{Peak index match probabilities versus $D$ across various channel inference techniques based on positional information for \textbf{Exp OOD}. The results show that only the cDDIM method can generalize to out-of-distribution channel generation.}\label{figure:probabilitycomparison2.pdf}
\end{figure}

\smallskip \textbf{Analysis of LOS Peaks.} 
MSE or correlation-based metrics often overlook the sparse, dominant-path structure of mmWave channels, notably the LOS path, which typically dominates system performance \cite{Alkhateeb2018Deep, Heath2016mmWaveMIMO}. Even as discussed in Section~\ref{Problem Setup}, it is inappropriate to use MSE in our setup due to small-scale fading. The direction of the LOS path is important to avoid misleading performance gains in beamforming or compression \cite{Rappaport2013mmWave5G}.
Mathematically, the peak BS side index of channel  $\mathbf{H}_\mathrm{v}$, $i_{\mathbf{H}_\mathrm{v},\text{BS}}$ is defined as
\begin{equation}
i_{\mathbf{H}_\mathrm{v},\text{BS}} = \argmax_i \max_j \mathbf{H}_{\mathrm{v},ji}.
\end{equation}

Let's examine the difference between the peak BS side index of the augmented channel dataset $\{\widetilde{\mathbf{H}}_{\mathrm{v,aug},i}\}_{i=1}^{N_\mathrm{aug}}$, denoted as $i_{\widetilde{\mathbf{H}}_{\mathrm{v,aug}},\text{BS}}$ and the reference channel dataset $\{{\mathbf{H}}_{\mathrm{v,aug},i}\}_{i=1}^{N_\mathrm{aug}}$, denoted as $i_{{\mathbf{H}}_{\mathrm{v,aug}},\text{BS}}$. We will compare these indices individually for channels in the same position. 

We are interested in the distribution of the peak index difference between the peak BS side indices of the two sets, $D = \left| \left| i_{\widetilde{\mathbf{H}}_{\mathrm{v,aug}},\text{BS}}- i_{{\mathbf{H}}_{\mathrm{v,aug}},\text{BS}} \right| \right|$. Ideally, if the augmented channel always predicts the LOS peak correctly, then  $D=0$. However, since the augmented channel may have some errors compared to the reference channel, $D$ can be non-zero. We plot the cumulative distribution function (CDF) of $D$ to evaluate how well each augmentation method predicts the location of the LOS peak.

In \textbf{Exp ID}, where the training and test sets come from the same distribution, 
both cDDIM and the One‑Shot U‑Net accurately estimate the peak beam index, 
exhibiting nearly identical performance. Specifically, their CDF at \(D=0\) 
is around 0.2—indicating an exact match 20\% of the time—and their CDF at \(D=2\) 
reaches about 0.7, so the difference from the true peak is within two indices 
70\% of the time. Remarkably, this even surpasses the accuracy of simply selecting 
the closest UE location from the training set, demonstrating that these models 
effectively interpolate channels between measured points. In contrast, cGAN and no‑DFT cDDIM perform no better than random guessing: cGAN collapses in mode diversity, while no‑DFT cDDIM fails to learn how the peak index varies with UE position.

Turning to \textbf{Exp OOD}, which employs a different (out‑of‑distribution) test region, all models degrade in accuracy. Nonetheless, cDDIM still 
maintains robust peak estimation (e.g., more than 50\% of its predictions lie within two indices of the true peak \((D \leq 2)\)), whereas one‑shot U‑net degenerates to cGAN‑level performance. The one‑shot U‑net is inherently more of a regression method that excels at interpolation within the original distribution, yet blurs the output for unseen locations. Conversely, cDDIM leverages its generative capability to capture the overall channel distribution, making it more adaptable even when encountering new, out‑of‑distribution positions. Hence, although both methods share the same U‑net backbone, using a diffusion-model framework to capture and augment the channel distribution proves crucial for strong performance. We do not include the ``closest in training set'' baseline here under distribution shifts.

\section{Application to Downstream Tasks}
\label{Downstream Applications and Evaluation}

This section presents two different downstream applications of our proposed amplified datasets. Using the dataset generated in Section \ref{Proposed Methods}, we aim to apply it to various data-driven solutions across different wireless communication tasks to determine if the amplified dataset yields better results. The advantage comes from the diffusion model's ability to produce better-interpolated datasets, and we can evaluate by the performance in downstream tasks.

Two different downstream tasks are (1) channel compression
and (2) site-specific beam alignment. The first experiment uses QuaDRiGa \cite{quadriga2023}, and the second experiment uses DeepMIMO \cite{alkhateeb:19} due to the nature of the experiments. Both experiments confirm that our cDDIM method performs effectively with statistically designed channels (as in QuaDRiGa) and with ray-tracing-based sparse channels (DeepMIMO) as well.

\smallskip \textbf{Baselines.} Several methods can be used to augment the dataset with channels. We consider (i) adding Gaussian noise and (ii) ChannelGAN \cite{Xia:22} as baseline methods.

\emph{Adding Gaussian noise}: We add 10 dB Gaussian noise to the channel matrix to amplify the dataset, similar to our method. The noise level is compared to the Frobenius norm of the channel. It is necessary to amplify the dataset significantly in size, so if the Gaussian noise level is too low, there is not much difference in the dataset amplification. Therefore, the noise level is empirically selected to make the amplification effective.

\emph{ChannelGAN}: ChannelGAN follows the structure of WGAN-GP \cite{WGAN}, consisting of two networks: a generator and a discriminator. The generator creates fake channels from a random latent vector while the discriminator determines whether the channels are real or fake. After training, the generator can make synthetic channel data to form an extensive training dataset, similar to our method. ChannelGAN work does not include positional data. Therefore, for every experiment, we amplify the dataset by ChannelGAN to 90,000 by sampling the channel from randomly sampled latent vectors.

\subsection{Channel Compression}
\label{Channel Compression}
\begin{figure}[!tbp]
  \centering
  \includegraphics[width=0.9\linewidth]{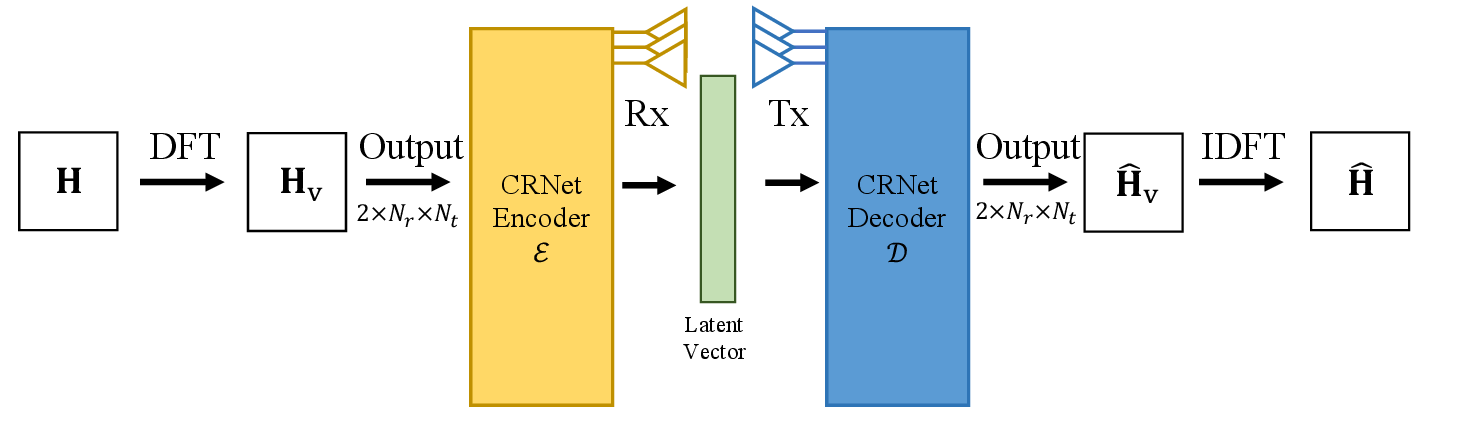}
  \caption{Diagram of CRNet aided downlink CSI feedback \cite{CRNet}}\label{figure:csicompression.pdf}
\end{figure}

\smallskip \textbf{Problem Setup.} In this work, we focus on improving the CSI feedback in MIMO systems. Specifically, we aim to evaluate the normalized mean square error (NMSE) of the reconstructed downlink (DL) CSI when using different dataset augmentation methods. The goal is to reduce the NMSE between the original and reconstructed CSI using minimal training data. For a visualization of the scenario, refer to Fig. \ref{figure:csicompression.pdf}.

Under the same training scheme, channel reconstruction
network (CRNet) \cite{CRNet} outperforms CsiNet \cite{CsiNet} and CsiNetPlus \cite{CsiNetPlus} with stable NMSE. Therefore, we chose CRNet for the downstream task. This work focuses solely on the feedback scheme, assuming ideal downlink channel estimation and uplink feedback. While the original work assumes a MISO FDD system, we assume a MIMO narrowband system, leading to different settings but a similar model structure due to the 2D sparse channel.

A CRNet consists of two deep neural networks: an encoder $\mathcal{E}$ and a decoder $\mathcal{D}$. First, we apply DFT to the channel matrix $\mathbf{H}$ to obtain its beamspace representation $\mathbf{H}_{\mathrm{v}}$. We then input the channel matrix $\mathbf{H}_{\mathrm{v}}$ into the encoder $\mathcal{E}$. Subsequently, we decode the latent vector using the decoder $\mathcal{D}$ and perform inverse DFT (IDFT) to reconstruct the channel matrix. The following formula \eqref{CRNet} explains the process, illustrated in Fig. \ref{figure:csicompression.pdf}.

\begin{equation}
\widehat{\mathbf{H}}_{\mathrm{v}}=\mathcal{D}\left(\mathcal{E}\left(\mathbf{H}_{\mathrm{v}},\Theta_{\mathcal{E}}\right), \Theta_{\mathcal{D}}\right)
\label{CRNet}
\end{equation}

\smallskip \textbf{Neural Architectures.}
The generation of $\mathbf{H}_{\mathrm{v}}$ involves using a DFT to convert the spatial domain channel matrix $\mathbf{H}$ to the beamspace representation $\mathbf{H}_{\mathrm{v}}$. The encoder $\mathcal{E}$ processes the channel matrix $\mathbf{H}_{\mathrm{v}}$, treated as an input image of size $2 \times N_{\mathrm{r}} \times N_{\mathrm{t}}$, where $N_{\mathrm{r}}$ and $N_{\mathrm{t}}$ are the antenna dimensions. The input passes through two parallel paths—one with three serial convolution layers for high resolution and the other with a single $3 \times 3$ convolution layer for lower resolution. These outputs are concatenated and merged with a $1 \times 1$ convolution, followed by a fully connected layer that scales down the features to a latent vector whose size is reduced by the compression rate.

The decoder $\mathcal{D}$ then scales up and resizes the received feature vector, processes it through a convolution layer for rough feature extraction, and further refines it using two CRBlocks. Each CRBlock contains parallel paths with different resolutions, and their outputs are merged with a $1 \times 1$ convolution layer, incorporating residual learning through identity paths. The process is completed with an additional sigmoid layer for activation, as in \cite{CRNet}.

\smallskip \textbf{Methods.} Assume that an $N_{\text{train}}$ dataset is given, and we aim to boost this dataset to $N_{\text{train}} + N_{\text{aug}} = 90,000$. We consider several methods for dataset augmentation:
\begin{itemize}
\item \textbf{Reference channels} Naively using $N_{\text{train}}$ samples.
\item \textbf{Our Method (cDDIM)} Boosting $N_{\text{train}}$ with $N_{\text{aug}}$ channels using cDDIM.
\item \textbf{Adding Gaussian Noise} Boosting $N_{\text{train}}$ with $N_{\text{aug}}$ channels by adding Gaussian noise.
\item \textbf{ChannelGAN} Boosting $N_{\text{train}}$ with $N_{\text{aug}}$ channels using ChannelGAN.
\end{itemize}

The entire $N_{\text{train}} + N_{\text{aug}}$ dataset is used for training the CRNet.

\smallskip \textbf{Experiment Setup.} In our simulation specifications, the BS is configured with 32 antennas and each UE with 4 antennas. Performance is measured by the NMSE difference between the reconstructed channel $\widehat{\mathbf{H}}_{\mathrm{v}}$ and ${\mathbf{H}}_{\mathrm{v}}$.  NMSE is appropriate as a metric in this task, as in channel estimation more broadly, since the phase of the channel is critical.

For the creation of our samples from the environment, we started with $N_{\text{ref}} = 90,000$ samples for the reference channel and position pair dataset, and $N_{\text{test}} = 10,000$ samples for the test channel dataset from the QuaDRiGa simulator. We then randomly selected $N_{\text{train}} = 0.5K, 1K, 2K, $
$4K, 6K, 8K,$ and $10K$ channel and position pair samples from the $N_{\text{ref}}$ channels, and boosted them to a total of $N_{\text{train}} + N_{\text{aug}} = 90K$ using the cDDIM augmentation method explained in Section \ref{Proposed Methods}. We trained the CRNet with these $N_{\text{train}} + N_{\text{aug}} = 90,000$ samples and evaluated it on the $N_{\text{ref}} = 10,000$ test channel samples. We also used ChannelGAN and Gaussian noise augmentations as baselines to compare with the cDDIM augmentation, as explained above. The number of epochs was set to 500, and the Adam optimizer was used for training.

\begin{figure}[!tbp]
  \centering
  \includegraphics[width=0.9\linewidth, trim={30pt 190pt 57pt 200pt}, clip]{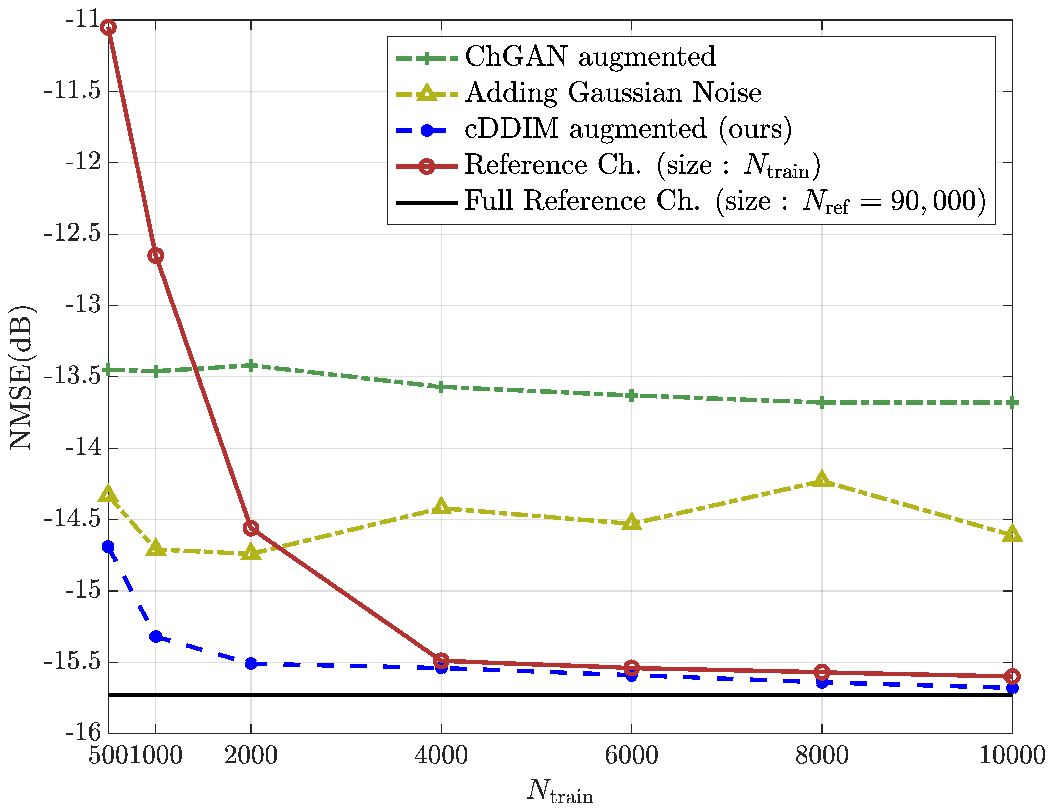}
  \caption{NMSE comparison of different augmentation methods for channel compression. cDDIM augmentation is the only method that achieves low NMSE for every $N_\text{train}$.}\label{figure:Channelcompressionresult.pdf}
\end{figure}

\smallskip \textbf{Evaluation Metric: NMSE.} Our evaluation metric is as follows,
\begin{equation}
\mathrm{NMSE}=\mathbb{E}\left\{\frac{\left\|\widehat{\mathbf{H}}_{\mathrm{v}}-{\mathbf{H}}_{\mathrm{v}}\right\|_F^2}{\left\|\mathbf{H}_{\mathrm{v}}\right\|_F^2}\right\}.
\end{equation}

When we train the model multiple times, the NMSE differs each time due to the gradient descent method yielding different local minima. We observed that the model with the smallest training loss also performed best on the test dataset. Therefore, each experiment was conducted five times, and the model with the best validation NMSE value was selected as the well-trained model.

\smallskip {\bf Results. } We summarize the results in Fig. \ref{figure:Channelcompressionresult.pdf}. The black line represents the NMSE when the CRNet is trained with $N_\mathrm{ref} = 90,000$ reference samples, serving as the lower bound for NMSE performance. 
The red line with $\texttt{o}$ markers shows NMSE when trained with varying $N_\mathrm{train}$ values. With only $N_\mathrm{train} = 500$, NMSE degrades by 5 dB compared to the black line.

However, when the dataset is augmented to 
$N_\mathrm{train} + N_\mathrm{aug} = 90,000$ using cDDIM (blue line with
\begin{tikzpicture}[baseline=-0.5ex]
    \node[circle, draw, fill=black, inner sep=0pt, minimum size=0.5ex] {};
\end{tikzpicture} markers), NMSE remains within 1 dB of the black line, despite starting with only 0.5\% of the total data. This shows that cDDIM allows performance close to the lower bound with just 0.5\% of the dataset. Although both methods achieve near-lower-bound performance with 4,000 samples, at 500 samples (1/8 of 4,000) reference dataset lags by 5 dB while cDDIM is only 1 dB away from the lower bound.
Other methods, like Gaussian noise (yellow line with
\begin{tikzpicture}[baseline=-0.5ex]
    \node[regular polygon, regular polygon sides=3, draw, minimum size=0.7em, inner sep=0pt, rotate=0] {};
\end{tikzpicture}
markers) and ChannelGAN (green line with 
$\texttt{+}$ markers) show consistently higher NMSE by 1-2 dB, regardless of the dataset size.

\smallskip {\bf Interpretation. } Adding Gaussian noise may increase robustness but does not introduce new information. We also conjecture that ChannelGAN-based augmentation introduces bias, as it randomly generates channels rather than interpolating between positions. This can result in matrices that do not adequately represent the diversity of the dataset, potentially missing the necessary interpolated channels needed to address data scarcity.

\subsection{Site-specific Beam Alignment Engine (BAE)}

\begin{figure}[!tbp]
  \centering
  \includegraphics[width=0.6\linewidth]{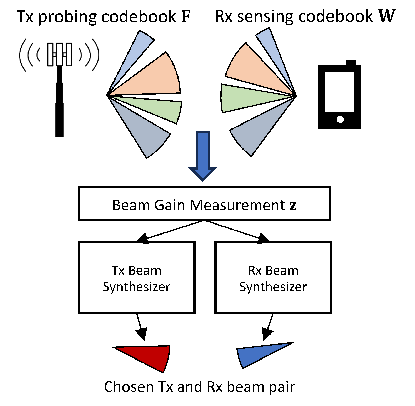}
  \caption{Illustration of the beam alignment engine based on site-specific probing \cite{heng:23}.}\label{figure:sitespecific.pdf}
\end{figure}

\smallskip \textbf{Problem Setup.} The deep learning-based grid-free beam alignment engine (BAE) introduced in \cite{heng:23} aims to learn transmit (Tx) probing beams tailored to the overall channel distribution. Initially, the BS sweeps a probing beam matrix to gather channel information,
\begin{equation}\label{eq:DL_pilot}
    \mathbf{Y} = \sqrt{P_t} \mathbf{W}^H \mathbf{H} \mathbf{F} \mathbf{s} + \mathbf{W}^H \mathbf{N},
\end{equation}
where $P_t$ is the transmit power, $\mathbf{W}$ and $\mathbf{F}$ are the receive and transmit beamforming matrices, $\mathbf{H}$ is the channel matrix, $\mathbf{s}$ is the probing symbol, and $\mathbf{N}$ is the noise matrix. 

Then, all connected UEs measure and report the received power of the probing signal. These $N_{\text{probe}}$ probing beam measurements become inputs to a multi-layer perceptron (MLP):
\begin{equation}
    \mathbf{z} = \left[ |[\text{diag}(\mathbf{Y})]_1|^2 \cdots |[\text{diag}(\mathbf{Y})]_{N_{\text{probe}}}|^2 \right]^T,
\end{equation}
that outputs the final selected beams $\mathbf{v}_r$ and $\mathbf{v}_t$ by the neural network. 

End-to-end deep learning jointly trains the final beam selector and probing beam matrix when pretraining. Refer to Fig. \ref{figure:sitespecific.pdf} for a scenario visualization.

\smallskip \textbf{Neural Architectures.} The model consists of two kinds of deep learning modules: the Complex NN Module and the Beam Synthesizer Module. The first module, the Complex NN Module, generates the Tx probing beams and Rx sensing beams. This module includes trainable parameters for the Tx complex probing weights and Rx complex sensing weights. These weights are then element-wise normalized. The first module outputs the probing beams and sensing beams, which are used to measure the received signals.


The Complex NN Module's measurements are fed into the Tx and Rx Beam Synthesizers, each consisting of dense layers with ReLU activation and batch normalization for stability. The final layer outputs the real and imaginary parts of the beamforming weights, scaled for effective beam alignment. For details, see \cite{heng:23}.

We ignore the additional initial access term $\mathcal{U}_{\text{IA}}$ explained in \cite{heng:23}  as it depends on large-scale fading, which our cDDIM model does not capture due to channel normalization. Thus, we set the initial access loss to 0 and normalize the synthetic channels to maintain a constant Frobenius norm for site-specific beam alignment.

\smallskip \textbf{Methods.} We start with a smaller initial dataset of $N_\text{train} = 100$, and augment it to $N_\text{train} + N_\text{aug} = 400,000$.  using reference channels, cDDIM, Gaussian noise, and ChannelGAN (as described in Section \ref{Channel Compression}). The augmented dataset is then used to train the BAE. We also include other baselines for comparison.

\begin{itemize}
    \item \textbf{MRC+MRT (Upper bound)}: No codebook; BS uses MRT, UE uses MRC. Theoretical upper bound via eigendecomposition is not achievable with unit-modulus constraint.
    \item \textbf{DFT+EGC}: BS has a codebook; exhaustively tries beams, UE uses EGC. Selects the best pair, assuming no noise, which is better than Genie DFT due to UE freedom.
    \item \textbf{Genie DFT}: Genie selects optimal beams in BS and UE codebooks, equivalent to an exhaustive search with zero noise.
    \item \textbf{Exhaustive Search}: Measures all beam pairs in BS and UE codebooks, selects highest received signal power, may not maximize SNR due to noise.
\end{itemize}

These are conventional methods for beam selection. The above methods are known to have much higher time complexity than site-specific beamforming \cite{heng:23}, but they serve as good benchmarks to see how our method performs in terms of beamforming gain. We focus on the average SNR performance of the above four methods to compare with the trained BAE using the boosted dataset.

\begin{table}
\small
\centering
  \caption{Site-specific beamforming simulation parameters}\label{table:site-specific}
    \begin{tabular}{| c | c |}
    \hline
    
    Simulation Environment & DeepMIMO \\ \hline
    Scenario Name & Outdoor 1 Blockage \\ \hline
    UE Pre-Augmentation $N_\text{train}$ & 100 \\ \hline
    UE Augmented Samples $N_\text{aug}$ & 239,900 \\ \hline
    UE Inference Samples $N_\text{test}$ & 80,000 \\ \hline
    BS Antenna & $4\times4$ UPA \\ \hline
    BS Codebook Size & $8\times8 = 64$ \\ \hline
    UE Antenna & $2\times2$ UPA \\ \hline
    UE Codebook Size & $4\times4 = 16$ \\ \hline
    Training Epochs & 500 \\ \hline
    Carrier Frequency & \makecell{28 GHz} \\ \hline
    Bandwidth ($B$) & 100 MHz\\ \hline
    BS Power & 35 dBm\\ \hline
    Noise Power ($\sigma^2$) & -81 dBm \\ \hline
    \end{tabular}
\end{table}

\smallskip \textbf{Experiment Setup.} In this experiment, we use the DeepMIMO dataset \cite{alkhateeb:19} to ensure site-specificity, focusing on a 28 GHz outdoor blockage scenario with two streets, an intersection, and three added surfaces as reflectors and blockers. The BS uses 16 uniform planar array (UPA) arrays, and the UE uses 4 UPA arrays.

We started with $N_\text{ref} = 240,000$ samples for the reference channel and position pair dataset, and $N_\text{test} = 80,000$ samples for the test channel dataset from the DeepMIMO simulator. Then, we randomly sampled $N_\text{train} = 100$ samples from the reference channel dataset and boosted them to $N_\text{train} + N_\text{aug} = 240,000$ using the cDDIM augmentation method. The BAE was trained on these 240,000 samples and evaluated on the 80,000 test channels, with ChannelGAN and Gaussian noise augmentations used as baselines.

\begin{figure}[!tbp]
  \centering
  \includegraphics[width=0.9\linewidth, trim={0pt 0pt 0pt 0pt}, clip]{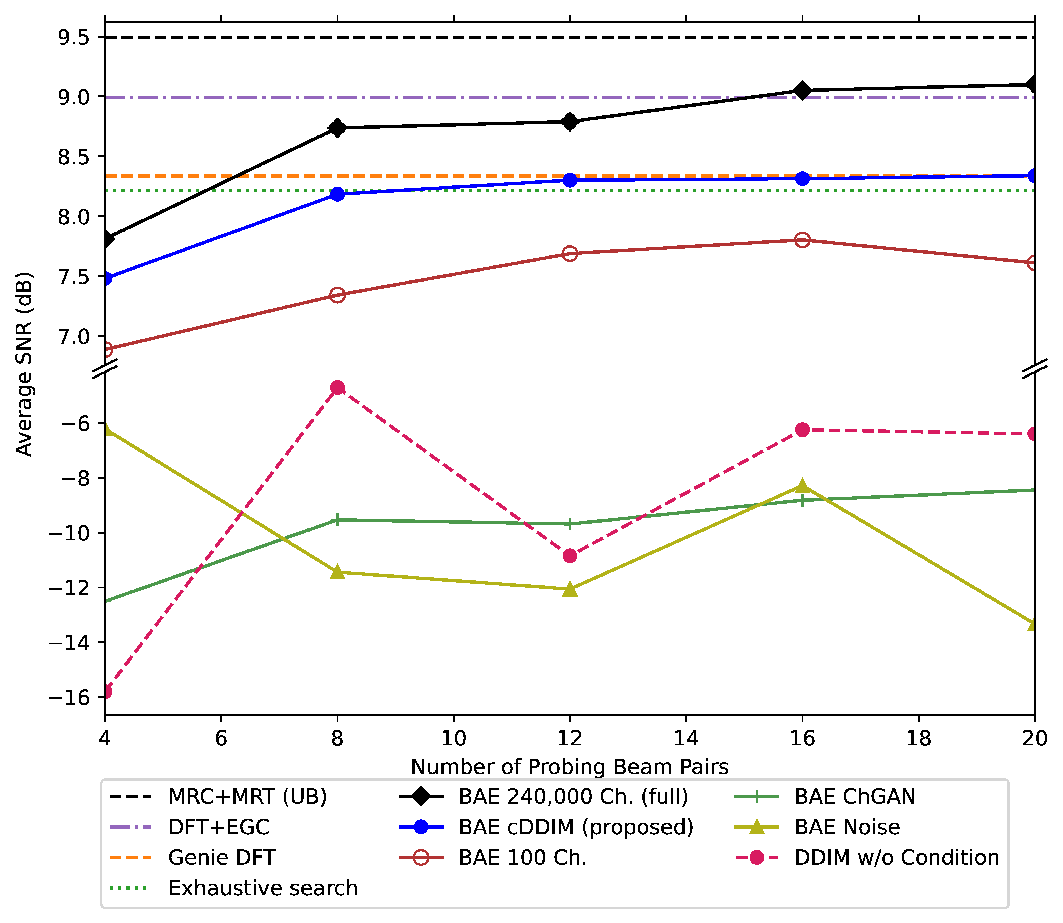}
  \caption{Average SNR of synthesized beam vs. number of beams across various beamforming and augmentation techniques. cDDIM-based augmentation is the only method that shows enhanced beamforming SNR.}\label{figure:graph_result_adjusted.pdf}
\end{figure}

\smallskip \textbf{Evaluation Metric: Average SNR.}
We evaluate BAE performance by calculating the average SNR of the synthesized beams, defined as the ratio of the selected beam's power to noise power, averaged over all test samples:
\[
\text{Average SNR} = \frac{1}{N_{\text{test}}} \sum_{i=1}^{N_{\text{test}}} \frac{P_t |\mathbf{v}_r^H \mathbf{H}_i \mathbf{v}_t|^2}{\sigma^2},
\]
where $N_{\text{test}}$ is the number of test samples, $P_t$ is the transmit power, $\mathbf{v}_r$ and $\mathbf{v}_t$ are the receive and transmit beamforming vectors, $\mathbf{H}_i$ is the channel matrix for the $i$-th test sample, and $\sigma^2$ is noise power. 

We vary the number of probing beam pairs, $N_{\text{probe}}$, which determines the number of columns in the combiner $\mathbf{W}$ and the precoder $\mathbf{F}$. More $N_{\text{probe}}$ leads to better estimates of $\mathbf{v}_r$ and $\mathbf{v}_t$.

\smallskip \textbf{Results.}
As illustrated in Fig. \ref{figure:graph_result_adjusted.pdf}, the BAE trained with the dataset augmented by the cDDIM method (blue line with \begin{tikzpicture}[baseline=-0.5ex]
    \node[circle, draw, fill=black, inner sep=0pt, minimum size=0.5ex] {};
\end{tikzpicture} markers) shows a significantly higher average SNR of the synthesized beam compared to the BAE trained with datasets augmented by ChannelGAN (green solid line with $\texttt{+}$ markers) or Gaussian noise (yellow solid line with \begin{tikzpicture}[baseline=-0.5ex]
    \node[regular polygon, regular polygon sides=3, draw, minimum size=0.7em, inner sep=0pt, rotate=0] {};
\end{tikzpicture} markers). 
The SNR gap between cDDIM and the full dataset (black line) is about 1 dB.

Using more than 16 beams in cDDIM-based augmentation consistently outperforms both exhaustive search (green dotted line) and the Genie DFT case (orange dashed line). Deep learning-based methods with grid-free beams outperform DFT beams, but Gaussian noise and ChannelGAN fail to improve average SNR consistently as beams increase, including undesirable interpolations.
ChannelGAN and adding noise both exhibit significantly worse SNR, ranging from -13 dB to -7 dB, indicating that the power of the selected beam is lower than environmental noise, making them ineffective for beamforming. The unconditioned cDDIM method (pink dashed line) yields poor average SNR on par with Gaussian noise and ChannelGAN augmentation, highlighting the necessity of conditioning on UE positions.

\section{Conclusion}
\label{Conclusion}

We proposed a novel framework for augmenting wireless channel datasets using a conditional diffusion model. We demonstrate that it is possible to significantly enhance the realism and applicability of synthetic datasets, which are crucial for training robust deep-learning models for wireless applications. 

Future work could apply contrastive learning to force the dominant LoS peak to vary smoothly with user position, thereby making the synthesized channels more realistic. Also, our experiments indicate that the model is sensitive to user-location noise; even small inaccuracies in location data—due to, for instance, user mobility—can degrade performance. This underscores the need to explore strategies for enhancing robustness to positional uncertainty, such as incorporating velocity or additional mobility-related factors, in future work.
Additionally, while we introduced random UE orientation in our experiments, the impact of varying antenna radiation patterns remains important for extending the proposed approach.
Future work could move the diffusion model into the multipath-parameter domain, enabling geometry-agnostic, interpretable channels—though path-label scarcity remains a key challenge. Likewise, it remains open whether a retrained cDDIM backbone—with spherical or hybrid waves—can also serve near-field (sub-6 m) or THz links.

\bibliographystyle{ieeetr}
\begingroup
\bibliography{AZREF}
\endgroup

\end{document}